\documentclass[a4paper]{scrartcl}
\usepackage[utf8]{inputenc}
\usepackage[T1]{fontenc}
\usepackage{lmodern}
\usepackage[english]{babel}
\usepackage{microtype}
\usepackage{csquotes}
\usepackage[includeheadfoot,margin=2.5cm]{geometry}

\usepackage[dvipsnames]{xcolor}
\usepackage{graphicx}
\usepackage{amsfonts}
\usepackage{amssymb}
\usepackage{amsmath}
\usepackage{amsthm}
\usepackage{enumerate}
\usepackage{multirow}

\usepackage[backend=biber,style=numeric,maxnames=5,sortcites]{biblatex}
\usepackage[colorlinks]{hyperref}
\hypersetup{
  linkcolor=BrickRed,
  citecolor=Green,
  urlcolor=NavyBlue,
  menucolor=BrickRed
}

\usepackage[noabbrev,capitalize]{cleveref}
\crefname{equation}{}{}

\usepackage{authblk}

\usepackage{orcidlink}

\newtheorem{theorem}{Theorem}
\newtheorem{observation}[theorem]{Observation}
\newtheorem{corollary}[theorem]{Corollary}
\newtheorem{lemma}[theorem]{Lemma}

\theoremstyle{definition}

\newtheorem{example}[theorem]{Example}

\newcommand{\R}{\mathbb{R}}
\newcommand{\N}{\mathbb{N}}

\newcommand{\set}[1]{\{#1\}}
\newcommand{\fromto}[2]{\set{#1,\dots,#2}}
\newcommand{\El}{E^\ell}
\newcommand{\Ef}{E^f}
\newcommand{\Pst}{\mathcal{P}_{st}}
\newcommand{\BSP}{\textsc{Bilevel Shortest Path}}
\newcommand{\BSPweakDir}{\textsc{BSP-Weak-Dir}}
\newcommand{\BSPweakUndir}{\textsc{BSP-Weak-Undir}}
\newcommand{\BSPstrongDir}{\textsc{BSP-Strong-Dir}}
\newcommand{\BSPstrongUndir}{\textsc{BSP-Strong-Undir}}
\newcommand{\kcycle}[1]{\textsc{Shortest-${#1}$-Cycle}}
\newcommand{\MinMaxHam}{\textsc{Min-Max-Ham}}
\newcommand{\HamPath}{\textsc{Hamiltonian Path}}
\newcommand{\SP}{\textsc{Shortest Path}}
\newcommand{\VDP}{\textsc{Vertex-Disjoint Paths}}
\newcommand{\IndSet}{\textsc{Independent Set}}
\newcommand{\SAT}{\textsc{3SAT}}
\newcommand{\QSAT}{$\exists\forall$-\textsc{DNF-3SAT}}

\DeclareMathOperator*{\argmin}{arg\,min}
\DeclareMathOperator*{\OPTw}{\textsc{OPT}_\text{weak}}
\DeclareMathOperator*{\OPTs}{\textsc{OPT}_\text{str}}
\DeclareMathOperator*{\OptFOLs}{\textsc{OPT}^{\text{\,f}}_\text{str}}
\DeclareMathOperator*{\OptFOLw}{\textsc{OPT}^{\text{\,f}}_\text{weak}}
\DeclareMathOperator*{\ProblemFOLs}{\textsc{FOL}_\text{str}}
\DeclareMathOperator*{\ProblemFOLw}{\textsc{FOL}_\text{weak}}

\begin{document}

\title{On the Complexity of the \\ Bilevel Shortest Path Problem}

\author[1]{Dorothee Henke\textsuperscript{\orcidlink{0000-0001-9190-642X}}}
\author[2]{Lasse Wulf\textsuperscript{\orcidlink{0000-0001-7139-4092}}}

\affil[1]{Chair of Business Decisions and Data Science, University of Passau, Germany. \texttt{dorothee.henke@uni-passau.de}}	
\affil[2]{Section for Algorithms, Logic and Graphs, Technical University of Denmark, Kongens Lyngby, Denmark. \texttt{lawu@dtu.dk}}

\date{\vspace{-2\baselineskip}}

\maketitle

\begin{abstract}
  We introduce a new bilevel version of the classic shortest path problem and completely characterize its computational complexity with respect to several problem variants. In our problem, the leader and the follower each control a subset of the edges of a graph and together aim at building a path between two given vertices, while each of the two players minimizes the cost of the resulting path according to their own cost function. We investigate both directed and undirected graphs, as well as the special case of directed acyclic graphs. Moreover, we distinguish two versions of the follower's problem: Either they have to complete the edge set selected by the leader such that the joint solution is exactly a path, or they have to complete the edge set selected by the leader such that the joint solution is a superset of a path. In general, the bilevel problem turns out to be much harder in the former case: We show that the follower's problem is already NP-hard here and that the leader's problem is even hard for the second level of the polynomial hierarchy, while both problems are one level easier in the latter case. Interestingly, for directed acyclic graphs, this difference turns around, as we give a polynomial-time algorithm for the first version of the bilevel problem, but it stays NP-hard in the second case. Finally, we consider restrictions that render the problem tractable. We prove that, for a constant number of leader's edges, one of our problem variants is actually equivalent to the shortest-$k$-cycle problem, which is a known combinatorial problem with partially unresolved complexity status. In particular, our problem admits a polynomial-time randomized algorithm that can be derandomized if and only if the shortest-$k$-cycle problem admits a deterministic polynomial-time algorithm.
\end{abstract}


\bigskip

\noindent\textbf{Funding:} Lasse Wulf acknowledges support by the Austrian Science Fund (FWF): W1230, and the Carlsberg Foundation CF21-0302 ``Graph Algorithms with Geometric Applications''.



\newpage


\section{Introduction} \label{sec:intro}

In \emph{bilevel optimization},
two decision makers,
called the \emph{leader} and the \emph{follower},
solve their optimization problems
in a hierarchical order:
First, the leader selects their solution
and second, the follower solves their own problem,
which is parameterized in the leader's decision.
The follower's decision variables may in turn appear in the leader's objective function
(or even in the leader's constraints).
Therefore, the leader has to anticipate the follower's reaction
already when making their own decision.
This usually makes bilevel optimization problems hard to solve.
The field of bilevel optimization currently receives a lot of interest,
see e.g.\ \cite{ColsonMarcotteSavard2007,DempeEtAl2015,Dempe2020,BeckSchmidt2021} for general overviews and introductions.

Less attention has been devoted to bilevel \emph{combinatorial} optimization problems
and the analysis of their computational complexity.
Well-known complexity results
include the fact that
bilevel optimization problems
in which both objective functions and all constraints are linear
are (strongly) NP-hard in case of continuous variables
and $\Sigma^p_2$-hard in case of binary variables~\cite{Jeroslow1985,HansenJaumardSavard1992}.
The class $\Sigma^p_2$ belongs to the polynomial-time hierarchy
and is a natural framework for problems involving two decision stages.
We refer to~\cite{Woeginger2021}
for an accessible introduction to this complexity class
and to~\cite{GrueneWulf2024}
for a study of many $\Sigma^p_2$-complete problems
and techniques to obtain $\Sigma^p_2$-completeness proofs.
An important reason for studying this complexity class is that
$\Sigma^p_2$-hard problems cannot be expected to be representable as a polynomial-size integer linear program
(unless $\mathrm{NP} = \Sigma^p_2$)
and therefore require the development of new solution techniques~\cite{Woeginger2021}.

The aim of this work is to investigate the computational complexity
of a specific bilevel combinatorial optimization problem.
It is based on the classic {\SP} problem with nonnegative edge costs.
In our {\BSP} problem,
the edge set of a directed or undirected graph
is partitioned into edges owned by the leader and edges owned by the follower,
and each of the two players has their own cost function on all edges of the graph.
Their common goal is to build a path between two designated vertices,
where first the leader selects some of their edges
and then the follower completes the solution to a path by adding some of their edges.
Each of the two players aims at minimizing the total costs
of the whole resulting path with respect to their own cost function.
For a more formal definition of the problem,
see \cref{sec:BSP-def}.
This problem has not appeared in the literature before,
but seems very natural
in view of the types of bilevel combinatorial optimization problems
current research is interested in,
see \cref{sec:literature}.

One can distinguish several variants of the {\BSP} problem:
The graph may be directed or undirected, or even directed acyclic.
Moreover,
when the follower completes the partial path selected by the leader,
one of the following two options can be seen as feasible:
Either all edges selected by the leader
have to be contained in the final path,
or the follower might be allowed to ignore some of the edges selected by the leader
if this is cheaper for them,
i.e.\ the joint solution is allowed to be a superset of a path.
We call the former the \emph{strong path completion} variant
and the latter the \emph{weak path completion} variant.
In total,
this gives six problem variants.
For each of them,
we analyze the computational complexity of the leader's problem and of the follower's problem.

\textbf{Our results.}
Our complexity results are summarized in \cref{tab:complexity-results}. We completely classify each problem variant into one of the three complexity classes P (i.e.\ solvable in polynomial time), NP-complete, or $\Sigma^p_2$-complete (i.e.\ even harder than NP-complete).
As expected in bilevel optimization, the leader's problem is always at least as difficult as the follower's problem. 
Another pattern that seems to emerge from our results is that the variant of strong path completion is usually more difficult than the variant of weak path completion.
This pattern is also not surprising, since already the follower's problem in the strong path completion variant involves finding a path through 
pre-specified edges, a task reminiscent of the {\HamPath} problem. 
Note however that, in the case of directed acyclic graphs, the pattern actually reverses, and the strong path completion variant is easier than the weak path completion variant. 
This could be deemed surprising.

We highlight in particular our result that the leader's problem 
in the strong path completion variant is even $\Sigma^p_2$-complete (\cref{thm:BSP-strong-sigma-2-complete}) because, in general, relatively few $\Sigma^p_2$-completeness proofs for combinatorial bilevel problems seem to be known.
Our proof is based on recent advances in the understanding of $\Sigma^p_2$-completeness in the context of min-max combinatorial optimization \cite{GrueneWulf2024}. 
Note however that, compared to the proofs in~\cite{GrueneWulf2024}, many nontrivial adaptions are still necessary for the bilevel setting of this paper. 
More details are given in \cref{subsec:sigma-2-completeness}.

Most of our hardness results can be further strengthened to inapproximability results, 
i.e.\ the corresponding problems are not approximable in polynomial time, unless $\mathrm{P} = \mathrm{NP}$. This is further discussed in \cref{sec:inapproximability}.

Finally, since most results in \cref{tab:complexity-results} are negative, we consider restrictions that make the {\BSP} problem tractable.
We consider the case where there are only constantly many edges that the leader can influence (denoted by $\lvert \El \rvert = O(1)$).
In this special case, we show that many variants can be solved efficiently, see \cref{tab:complexity-results-few-leaders-edges}. 
Interestingly, the strong path completion variant of the problem in undirected graphs turns out to be not so straightforward. 
We show that it is equivalent to the \kcycle{k} problem, 
which is a known combinatorial problem with partially unresolved complexity status \cite{bjorklund2012shortest}. 
This implies the following: 
For polynomially bounded edge costs, our problem admits a polynomial-time randomized algorithm 
that can be derandomized if and only if the \kcycle{k} problem can be solved in deterministic polynomial time (for constant $k$ and polynomially bounded edge weights).
For general edge weights, our problem can be efficiently solved if and only if the \kcycle{k} problem can be efficiently solved. A more detailed explanation of this equivalence is given in \cref{sec:few-leader-edges}.

\begin{table}
    \centering
    \begin{tabular}{ll|lll}
      && undirected & directed & directed acyclic\\
      \hline
      \multirow{2}{4em}{weak}
      & leader & NP (\cref{thm:BSP-weak-NP-complete}) & NP (\cref{thm:BSP-weak-NP-complete}) & NP (\cref{thm:BSP-weak-NP-complete}) \\
      & follower & P (\cref{lem:weak-followers-problem}) & P (\cref{lem:weak-followers-problem}) & P (\cref{lem:weak-followers-problem}) \\
      \hline
      \multirow{2}{4em}{strong}
      & leader & $\Sigma^p_2$ (\cref{thm:BSP-strong-sigma-2-complete}) & $\Sigma^p_2$ (\cref{thm:BSP-strong-sigma-2-complete}) & P (\cref{thm:BSP-strong-acyclic-P}) \\
      & follower & NP (\cref{thm:BSP-strong-follower-NP-complete}) & NP (\cref{thm:BSP-strong-follower-NP-complete}) & P (\cref{lem:BSP-strong-follower-acyclic}) \\

    \end{tabular}
    \caption{Complexity results obtained in this paper for the weak and the strong path completion variant of the {\BSP} problem in undirected, directed, and directed acyclic graphs. Problems are either shown to be $\Sigma^p_2$-complete, NP-complete, or solvable in polynomial time.}
    \label{tab:complexity-results}
\end{table}

\begin{table}
    \centering
    \begin{tabular}{l|ll}
      & undirected & directed \\
      \hline
      weak & P (\cref{obs:few-leader-edges-weak}) & P (\cref{obs:few-leader-edges-weak}) \\
      strong & equivalent to \kcycle{k} (\cref{thm:equivalence-k-cycle}) & NP (\cref{thm:few-leader-edges-strong-NP-complete})
    \end{tabular}
    \caption{Complexity results for the {\BSP} problem in the special case where $\lvert \El \lvert = O(1)$.}
    \label{tab:complexity-results-few-leaders-edges}
\end{table}

\subsection{Related literature} \label{sec:literature}

While the {\BSP} problem
that we study
has not been present in previous articles,
other bilevel optimization problems
that are based on the {\SP} problem
have been considered before.
They can be associated to the following three common types of bilevel combinatorial optimization problems.

A prominent example falls into the class of \emph{price-setting problems};
it has been introduced as a \emph{toll optimization problem}
and shown to be NP-hard in \cite{LabbeMarcotteSavard1998}.
Here the follower models one or several drivers in a network,
each solving a shortest path problem,
and the leader can charge tolls as additional costs for the follower on some of the edges.
The leader's goal is to maximize their profit from the tolls,
but should not discourage the follower from using the tolled edges
by making them too expensive.
This problem and similar problems
have also been studied in, e.g.\ \cite{MarcotteSavard2005,vanHoesel2008,LabbeViolin2013}.

Also in \emph{(partial) inverse problems},
the follower solves some basic problem,
such as a shortest path problem,
and the leader can influence the follower's objective function.
However,
the leader's aim is now that
some given (partial) solution to the follower's problem
becomes an optimal solution,
while minimizing the adjustments made to the objective function.
The inverse shortest path problem is solvable in polynomial time~\cite{AhujaOrlin2001},
but the partial inverse shortest path problem is NP-hard~\cite{LeyMerkert2024}.

A third class of bilevel combinatorial optimization problems
is the one of \emph{interdiction problems}.
In an interdiction problem,
the two decision makers have opposite objective functions:
The follower again solves some basic problem,
such as a shortest path problem,
and the leader can interdict some of the follower's resources,
e.g.\ by removing some edges,
in order to harm the follower as much as possible.
Interdiction versions of the shortest path problem
have been shown to be NP-hard and considered in many publications,
see e.g.\ \cite{BallGoldenVohra1989,MalikMittalGupta1989,IsraeliWood2002}.

Also the structure of our {\BSP} problem
can be seen to fall into a class of bilevel combinatorial optimization problems,
several of which
--~based on different classic combinatorial optimization problems,
but not the shortest path problem so far~--
have been investigated before.
However,
this problem class does not seem to have been given a name yet;
we propose to call it the class of \emph{partitioned-items bilevel problems}.
The characteristic feature of these problems is that
some ground set of items is partitioned into leader's items and follower's items.
Both decision makers together
build a feasible solution of some classic combinatorial optimization problem
by each selecting some items from their own set.
Moreover,
each player has their own objective function
that they optimize for the joint solution.

The simplest partitioned-items bilevel problem
is the \emph{bilevel selection problem}.
It is solvable in polynomial time
and has been studied,
together with robust versions of it,
in~\cite{Henke2024}.
A generalization of the bilevel selection problem
is a \emph{bilevel knapsack problem}.
It has been shown to be $\Sigma^p_2$-hard
and solvable in pseudopolynomial time~\cite{MansiAlvesCarvalhoHanafi2012,CapraraCarvalhoLodiWoeginger2014}.
The (partitioned-items) \emph{bilevel minimum spanning tree problem} is NP-hard,
as shown in~\cite{BuchheimHenkeHommelsheim2022},
and several variants of it have also been studied in~\cite{Gassner2002,ShiZengProkopyev2019,ShiProkopyevRalphs2023}.
In~\cite{BuchheimHenkeHommelsheim2022}, the authors also establish a connection between the special case of the bilevel minimum spanning tree problem where the follower controls only few edges and the shortest vertex-disjoint paths problem. This is similar to our equivalence result for the special case of the {\BSP} problem where the leader controls few edges and the \kcycle{k} problem.
Also the (partitioned-items) \emph{bilevel assignment problem} is set in a graph
and investigated in several variants,
all of which turn out to be NP-hard,
see~\cite{GassnerKlinz2009,FischerMulukWoeginger2022}.

\medskip

The remainder of the paper is structured as follows: In \cref{sec:BSP-def}, the formal definition of the {\BSP} problem is given. \cref{sec:BSP-weak} concerns the weak path completion variant. \cref{sec:BSP-strong} concerns the strong path completion variant. \cref{sec:few-leader-edges} discusses the special case $\lvert \El \rvert = O(1)$. Finally, we derive inapproximability results in \cref{sec:inapproximability} and conclude in \cref{sec:conclusion}.

\section{Problem formulation} \label{sec:BSP-def}

For the {\BSP} problem, we consider a simple graph $G = (V, E)$ which is either directed (in which case $E \subseteq V \times V$) or undirected (in which case $E \subseteq \binom{V}{2}$). 
The edge set~$E$ is partitioned into the \emph{leader's edges} $\El$ controlled by the leader, and the \emph{follower's edges}~$\Ef$ controlled by the follower, such that $E = \El \cup \Ef$ and $\El \cap \Ef = \emptyset$. 

We make the following assumptions w.l.o.g.: The graph~$G$ does not contain parallel edges 
(such edges can always be modeled with a simple graph by subdivision), and the leader and follower do not share any edge (a shared edge can be modeled by two parallel edges, where leader and follower control one edge each).

All paths and cycles in this paper are simple, i.e.\ they do not contain repeated vertices.
Moreover, we often use the same notation for a path $P$ and its set of edges.
Given two vertices~$s, t \in V$, we let $\Pst = \set{P \subseteq E : P \text{ is a path from $s$ to $t$}}$ denote the set of feasible solutions to the {\SP} problem. From now on, we always assume w.l.o.g.\ that $\Pst \neq \emptyset$, otherwise the problem is trivially infeasible.

The goal of the two players is now to jointly build an $s$-$t$-path by each selecting edges from their own edge set, $\El$ and $\Ef$, respectively, in a hierarchical order.
As explained in \cref{sec:intro}, we distinguish between the two problem variants where the solutions selected by the leader and by the follower either need to be combined to be precisely an $s$-$t$-path \emph{(strong path completion)} or need to only include an $s$-$t$-path as a subset \emph{(weak path completion)}.
In the first case, the follower is \enquote{forced} to use every edge from the leader's solution to build an $s$-$t$-path. In the second case, the follower may choose to ignore some edges selected by the leader in order to find an $s$-$t$-path that results in a minimal total cost for them.

An input instance of the {\BSP} problem is a tuple $I = (G,\El,\Ef, s,t,c,d)$,
consisting of a directed or undirected graph $G = (V,E)$
with $E = \El \cup \Ef$ and $\El \cap \Ef = \emptyset$,
vertices $s,t \in V$, a cost function $c : E \to \R_{\geq 0}$ for the leader, and a cost function $d : E \to \R_{\geq 0}$ for the follower. The strong path completion variant of the {\BSP} problem is then the optimization problem of finding the value

\begin{equation} \label{eq:BSP-strong}
  \begin{aligned}
    \OPTs(I) := \min\ & c(X \cup Y) \\
    \text{s.t. } & X \subseteq \El \\
    & Y \in
    \begin{aligned}[t]
      \argmin\ & d(Y') \\
      \text{s.t. } & Y' \subseteq \Ef \\
      & X \cup Y' \text{ is an $s$-$t$-path (i.e. }X \cup Y' \in \Pst),
    \end{aligned}
  \end{aligned}
\end{equation}
where we use the notation $f(Z) := \sum_{e \in Z} f(e)$, for an edge set $Z \subseteq E$ and a function $f \colon E \to \R$.

Similarly, the {\BSP} problem in the variant of weak path completion is defined as follows:
\begin{equation} \label{eq:BSP-weak}
  \begin{aligned}
    \OPTw(I) := \min\ & c(X \cup Y) \\
    \text{s.t. } & X \subseteq \El \\
    & Y \in
    \begin{aligned}[t]
      \argmin\ & d(Y') \\
      \text{s.t. } & Y' \subseteq \Ef \\
      & X \cup Y' \text{ inlcudes an $s$-$t$-path} \\
        & \text{(i.e. }\exists P \subseteq X \cup Y'\text{ with } P \in \Pst).
    \end{aligned}
  \end{aligned}
\end{equation}

For an example, which also illustrates the difference between the weak and the strong path completion variant, we refer to \cref{ex:weak-vs-strong} and \cref{fig:weak-vs-strong}.

We remark that some attention has to be paid to choices $X$ of the leader which make the follower's problem infeasible, i.e.\ the leader could theoretically  choose a set $X$ for which the follower is not able to find a set 
$Y'$ fulfilling all the constraints. 
However, if the follower's problem is infeasible, then the leader's constraint \enquote{$Y \in \argmin \dots $} cannot be satisfied. Therefore, such a choice $X$ is also infeasible for the leader's problem by definition, and the leader will never choose such a set.
Our assumption $\Pst \neq \emptyset$ implies that a feasible leader's solution always exists.

Observe that
the leader's objective function is evaluated on the joint solution $X \cup Y$,
while the follower only pays for their own edges $Y$.
However, the follower's objective $d(Y')$ could also be replaced by $d(X \cup Y')$
without changing the problem
because, from the follower's perspective,
the leader's choice $X$ is fixed and the value $d(X)$ is a constant.
Therefore,
it suffices to define the follower's costs $d$ on the follower's edges $\Ef$ instead of all edges $E$,
and we will usually assume w.l.o.g.\ that $d(e) = 0$ for all $e \in \El$.

Finally, in this paper, we focus on the \emph{optimistic setting}, which is a common assumption in bilevel optimization. In case of a tie, where the follower could choose multiple sets $Y$ which all have the same cost $d(Y)$ for the follower, but differ in their corresponding leader's cost $c(Y)$, we assume that the follower picks among them one which is cheapest for the leader.

Note that problems \cref{eq:BSP-strong} and \cref{eq:BSP-weak} are optimization problems. In order to describe their computational complexity, we introduce corresponding decision problems. 
Given an instance $I$ together with a threshold parameter $T \in \R_{\geq 0}$, the corresponding {\BSP} decision problem is to decide whether $\OPTs(I) \leq T$ (or $\OPTw(I) \leq T$, respectively) holds. 
We denote by \BSPstrongDir, \BSPstrongUndir, \BSPweakDir, {\BSPweakUndir} the four variants of the {\BSP} decision problem, each in the strong or weak path completion variant, and each restricted to directed or undirected graphs. 

Given an instance $I$ (which can be undirected or directed), we define $\OptFOLs(I, X)$ and $\OptFOLw(I, X)$ to be the optimal value of the follower's problem for this instance, given leader's solution $X \subseteq \El$. 
We also let $\ProblemFOLs(I, X)$ and $\ProblemFOLw(I, X)$ denote the computational problem to compute the follower's optimum, given the instance $I$ and the leader's solution $X$. 
We call this problem the \emph{follower's problem}, since it describes the task the follower has to solve from their perspective.

\subsection{General insights}
The following lemma states that the undirected variant of the {\BSP} problem can be seen as a special case of the directed variant. We will typically make use of the following consequence of the lemma: When the undirected case is NP-hard, the directed case is NP-hard as well. 
\begin{lemma}
    \label{lem:undir-special-case-of-dir}
    For both the strong and the weak path completion variant of the {\BSP} problem, the undirected case can be polynomially reduced to the directed case.
\end{lemma}
\begin{proof}
    \begin{figure}
        \centering
        \includegraphics[scale=1]{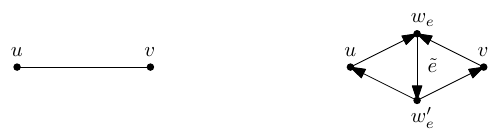}
        \caption{Illustration of the construction in the proof of \cref{lem:undir-special-case-of-dir} for the reduction of the undirected case to the directed case.}
        \label{fig:undir-spacial-case-of-dir}
    \end{figure}
    Given an instance $I = (G,\El,\Ef, s,t,c,d)$ of the undirected {\BSP} problem (in either the strong or the weak variant), consider the following modification of the graph $G$, as pictured in \cref{fig:undir-spacial-case-of-dir}: 
    For each undirected edge $e$, we introduce two new vertices $w_e$ and $w_e'$. 
    We replace the undirected edge $e = \set{u,v}$ by the five directed edges $(u, w_e)$, $(v, w_e)$, $(w_e, w'_e)$, 
    $(w'_e, u)$, and $(w'_e, v)$. 
    We define $\tilde e := (w_e, w'_e)$ and we let $\tilde e$ \enquote{inherit} the properties of $e$, i.e.\ $\tilde e$ is a leader's edge (follower's edge, respectively) if and only if $e$ is a leader's edge (follower's edge, respectively),
    and we let $c(\tilde e) := c(e)$ and $d(\tilde e) := d(e)$.
    For all the four other edges $f \neq \tilde e$, we let them be follower's edges with $c(f) = d(f) := 0$.
    It is now easily verified that the edge $\tilde e$ mimics the behavior of an undirected edge between $u$ and $v$. It can be used exactly once to go from either $u$ to $v$, or to go from $v$ to $u$. 
    Hence, the modified, directed, instance is equivalent to the original, undirected, instance.
  \end{proof}
  
As a small remark, we note that the natural approach to replace an undirected edge $\set{u, v}$ by the two edges $(u,v), (v,u)$ in the above lemma would not work: 
For example, if $\set{u,v}$ is a leader's edge, then in the new directed instance the leader would gain the ability to only unlock one of the two available directions. 
However, this ability is not present in the original instance, hence the two instances are not equivalent. 

\section{Weak path completion variant}
\label{sec:BSP-weak}

In this section, we are concerned with the weak path completion variant of the {\BSP} problem. We start with a few simple observations and then pinpoint the complexity of the two problem variants {\BSPweakDir} and {\BSPweakUndir}.

\begin{lemma}
\label{lem:weak-followers-problem}
For the weak path completion variant of the {\BSP} problem (both {\BSPweakDir} and {\BSPweakUndir}), the follower's problem is equivalent to a classic {\SP} problem and can therefore be solved in polynomial time.
\end{lemma}

\begin{proof}
  Let $(G, \El, \Ef, s, t, c, d)$ be an instance of the {\BSP} problem, where $G$ is either directed or undirected.
    Given some set $X \subseteq \El$ selected by the leader, the follower's problem in the weak path completion variant is to find a set $Y' \subseteq \Ef$ with minimal costs $d(Y')$ such that $X \cup Y'$ includes an $s$-$t$-path. 
    Since $\El$ and $\Ef$ are disjoint, this means that the follower can use all of the edges in $X$ without additional costs
    when looking for an $s$-$t$-path.
    However, all the edges in $\El \setminus X$ are forbidden to the follower. 
    Hence, if we introduce a cost vector
    \begin{align*}
        d'(e) := \begin{cases}
            d(e) & \text{if } e \in \Ef \\
            0 & \text{if } e \in X \\
            \infty & \text{if } e \in \El \setminus X,
        \end{cases}
    \end{align*}
    then the follower's problem is equivalent to finding a shortest $s$-$t$-path in the graph $G = (V, E)$ with respect to the costs $d'$. 
    Since all costs $d'(e)$ are nonnegative, in the case of directed graphs as well as in the case of undirected graphs, this problem can be solved using standard shortest path algorithms such as Dijkstra's algorithm~\cite{dijkstra}.
    
    In case multiple paths of the same minimal total cost exist, we assumed optimistically that the follower acts in favor of the leader, i.e.\ they select the path with smaller costs for the leader. This can be modeled in a standard way: Define $d''(e) :=  d'(e) + \varepsilon c(e)$ for all edges~$e \in E$ and some small enough $\varepsilon > 0$, for example $\varepsilon < (\min_{e \in E, d(e) \neq 0} d(e))/(1 + \sum_{e \in E} c(e))$. Then find a shortest path with respect to $d''$ (note that again $d'' \geq 0$).
  \end{proof}

Observe that
\cref{lem:weak-followers-problem} crucially relies on the assumption that we are in the optimistic setting (since this is necessary for $d'' \geq 0$). 
In fact, in the pessimistic setting, it can be easily seen that the follower's problem becomes NP-hard.
For example, if $d = 0$, then the follower would need to find a longest path in terms of $c$.

We next derive a structural property of optimal solutions of the {\BSP} problem, in case of weak path completion.
Informally speaking, the following \cref{lem:weak-optimal-solution-is-path} states that, in the weak path completion variant, we can w.l.o.g.\ assume that the optimal solution $X \cup Y$ is an $s$-$t$-path. 
At a first glance, this may be a bit confusing: If the defining feature of the weak path completion variant is that the follower is not restricted to form a path, why can we make such an assumption on the optimal solution? 
If the optimal solution is always a path anyway, what is the difference between the weak and strong path completion variant?
\begin{figure}
    \centering
    \includegraphics[scale=1]{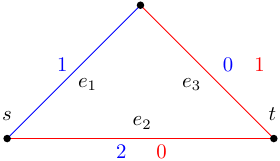}
    \caption{Instance discussed in \cref{ex:weak-vs-strong}. The leader's costs and edges are highlighted in blue, the follower's costs and edges are highlighted in red. Colored figure available online.}
    \label{fig:weak-vs-strong}
\end{figure}
The answer lies in the fact that \cref{lem:weak-optimal-solution-is-path} only holds for optimal leader's solutions, but does not necessarily hold for suboptimal leader's solutions.
This is illustrated in the following example, which shows that the weak and the strong path completion variants are indeed different.

\begin{example}
  \label{ex:weak-vs-strong}
Consider the instance depicted in \cref{fig:weak-vs-strong} for the weak path completion variant. We have $\El = \set{e_1}$ and $\Ef = \set{e_2, e_3}$. 
An optimal choice for the leader is to select $X = \emptyset$, to which the follower responds with $Y = \set{e_2}$. 
Note that $X \cup Y$ is an $s$-$t$-path in this case. 
If the leader selects the suboptimal solution $X = \set{e_1}$, then the optimal follower's solution is still $Y = \set{e_2}$. In this case, $X \cup Y$ is not an $s$-$t$-path anymore. 
Observe that, if we considered the strong instead of the weak path completion variant, the optimal leader's choice would now be $X = \set{e_1}$, to which the follower is forced to respond with $Y = \set{e_3}$.
\end{example}

The idea used for proving \cref{lem:weak-optimal-solution-is-path}
can be informally described as follows:
In case of weak path completion, given a (suboptimal) solution $X$ (such as $X = \set{e_1}$ in the example) that results in $X \cup Y$ not being an $s$-$t$-path, we can remove all edges from $X$ that the follower does not use for the path anyway. This does not change that $Y$ is an optimal follower's reaction to the modified leader's choice $X'$ (in the example $X' = \emptyset$) and does not increase the leader's costs.

\begin{lemma}
\label{lem:weak-optimal-solution-is-path}
Let $I = (G,\El,\Ef, s,t,c,d)$ be an instance of the {\BSP} problem in the weak path completion variant, where the graph $G$ is either directed or undirected (i.e.\ either {\BSPweakDir} or {\BSPweakUndir}). 
There exists at least one pair $(X, Y)$ with $X \subseteq \El$ and $Y \subseteq \Ef$ such that $X$ is an optimal choice for the leader's problem, $Y$ is an optimal choice for the follower's problem given~$X$, and $X \cup Y$ is an $s$-$t$-path (i.e.\ $X \cup Y \in \Pst$).
\end{lemma}

\begin{proof}
    Let $X \subseteq \El$ be defined such that among all optimal leader's solutions it has minimal cardinality. 
    After that, let $Y \subseteq \Ef$ be chosen such that, given $X$, among all optimal solutions to the follower's problem $\ProblemFOLw(I, X)$, it has minimal cardinality. 
    More precisely, $Y$ is chosen from the set $\set{Y' \subseteq \Ef : X \cup Y' \text{ includes an $s$-$t$-path}}$ in such a way that it minimizes in first priority $d(Y')$, it minimizes in second priority $c(Y')$ (optimistic assumption), and it minimizes in third priority $\lvert Y' \rvert$.
    
    By definition, $X$ and $Y$ are optimal choices for the leader and for the follower, respectively. 
    By the constraints of problem~\cref{eq:BSP-weak}, there is some $P \subseteq X \cup Y$ with $P \in \Pst$.
    We first claim that $Y \subseteq P$. Indeed, if there exists some edge $e \in Y \setminus P$, then it can be removed. 
    Since the cost functions $c$ and $d$ are nonnegative, we have $d(Y \setminus \set{e}) \leq d(Y)$, $c(Y \setminus \set{e}) \leq c(Y)$, and $\lvert Y \setminus \set{e} \rvert = \lvert Y \rvert - 1$, contradicting the definition of $Y$.
    
    Similarly, we claim that $X \subseteq P$. In order to prove this, suppose that there exists some edge $e \in X \setminus P$. We show that $X \setminus \set{e}$ is still an optimal choice for the leader, which contradicts the definition of $X$. Indeed, consider the sets
    \begin{align*}
        F_1 := \set{Y' \subseteq \Ef : X \cup Y' \text{ includes an $s$-$t$-path}} \text{ and}\\
        F_2 := \set{Y' \subseteq \Ef : (X \setminus \set{e}) \cup Y' \text{ includes an $s$-$t$-path}}
    \end{align*}
    of feasible solutions for the respective follower's problems $\ProblemFOLw(I, X)$ and $\ProblemFOLw(I, X \setminus \set{e})$. 
    Observe that we have $F_2 \subseteq F_1$, i.e.\ all possibilities to extend $X \setminus \set{e}$ to an edge set including an $s$-$t$-path can also be used to extend $X$, since the set of leader's edges and the set of follower's edges are disjoint.

    Observe furthermore that both $Y \in F_1$ and $Y \in F_2$. Since $Y$ was already an optimal element of $F_1$, it remains an optimal element in the smaller set $F_2$. 
    Since $c(X \cup Y \setminus\set{e}) \leq c(X \cup Y)$, it follows that $X \setminus \set{e}$ is also an optimal choice for the leader, arriving at a contradiction.

    In conclusion, we have proven that both $X, Y \subseteq P$. Since their union includes a path, this implies that indeed $X \cup Y = P$.
  \end{proof}

  We can now make the following statements about the relation between the weak and the strong path completion variant:

  \begin{observation}
    \label{obs:BSP-weak-strong-follower}
  Let $I$ be an instance of the {\BSP} problem (in a directed or an undirected graph)
  and $X$ a feasible leader's choice.
  Denote by $F_\text{str}$ and $F_\text{weak}$ the sets of feasible responses of the follower, given the leader's choice $X$.
  We then have $F_\text{str} \subseteq F_\text{weak}$
  and therefore $\OptFOLw(I, X) \le \OptFOLs(I, X)$.
\end{observation}

\begin{corollary}
  \label{cor:BSP-weak-strong-leader}
  For every instance~$I$ of the {\BSP} problem (in a directed or an undirected graph),
  we have $\OPTs(I) \leq \OPTw(I)$.
\end{corollary}

\begin{proof}
  Consider an optimal solution $(X, Y)$ for the weak path completion variant as in \cref{lem:weak-optimal-solution-is-path}, i.e.\ with $X \cup Y$ being an $s$-$t$-path. Then $Y$ is also a feasible and therefore (by \cref{obs:BSP-weak-strong-follower}) optimal response to~$X$ in the strong path completion variant.
  Therefore, the leader's costs resulting from choosing $X$ are the same in the strong case as in the weak one. An optimal solution in the strong case might be even cheaper than that.
\end{proof}

Finally, we prove that the {\BSP} problem, in the weak path completion variant, is hard:

\begin{theorem}
\label{thm:BSP-weak-NP-complete}
    Both problems {\BSPweakDir} and {\BSPweakUndir} are NP-complete. The directed variant {\BSPweakDir} is NP-complete even on directed acyclic graphs.
\end{theorem}
\begin{proof}
  First note that the problems {\BSPweakDir} and {\BSPweakUndir} are contained in NP.
  Indeed, given a solution $(X, Y)$, one can verify it to be feasible for the leader's problem (because the follower's problem can be solved in polynomial time by \cref{lem:weak-followers-problem}) and evaluate its leader's costs in polynomial time.
  
    For the hardness proof, we reduce from the NP-complete {\IndSet} problem \cite{garey1979computers}. 
    Given a graph~$G$ and a number $k \in \N$, the question is whether $\alpha(G) \geq k$. Here, $\alpha(G)$ denotes the size of the largest independent set in $G$.
    We first argue for the undirected case and then explain how the proof can be adapted to the directed acyclic case. Let $G = (V, E)$ and $k \in \N$ describe a given instance of the {\IndSet} problem and let $n := \lvert V \rvert$. We construct an instance $I$ of {\BSPweakUndir} such that $\OPTw(I) \leq 3n - k$ if and only if $\alpha(G) \geq k$.

    \begin{figure}
        \centering
        \includegraphics[scale=1]{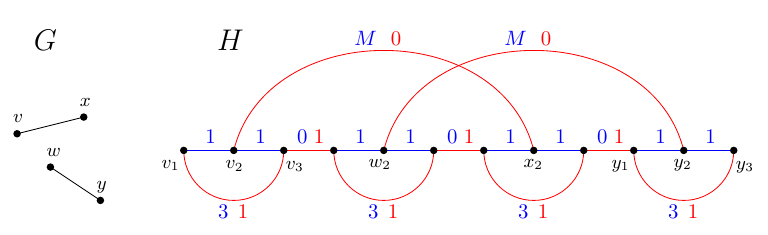}
        \caption{Construction used in the proof of \cref{thm:BSP-weak-NP-complete}. Leader's edges and costs are marked in blue, and follower's edges and costs are marked in red. Colored figure available online.}
        \label{fig:BSP-weak-NP-complete}
    \end{figure}
    The instance $I$ is sketched in \cref{fig:BSP-weak-NP-complete} and formally defined as follows.
    The underlying graph~$G'$ of the instance~$I$ has vertex set $V(G') := \bigcup_{v\in V} \set{v_1, v_2, v_3}$. The edge set $E(G') = \El \cup \Ef$ is split into four parts:
    \begin{itemize}
        \item All edges $\set{v_1,v_2}$ and $\set{v_2,v_3}$ for $v \in V$. 
        They are both leader's edges with leader's cost $c(e) = 1$ (and follower's cost $d(e) = 0$).
        \item All edges $\set{v_1,v_3}$ for $v \in V$. 
        They are follower's edges with $c(e) = 3$ and $d(e) = 1$.
        \item For some arbitrary order $v^{(1)},\dots,v^{(n)}$ of the vertices of $G$, the graph $G'$ contains all edges $\set{v_3^{(i)},v_1^{(i+1)}}$ for $i \in \fromto{1}{n-1}$. 
        They are follower's edges with $c(e) = 0$ and $d(e) = 1$.
        \item The edge set $E' := \set{ \set{u_2,v_2}$ : $\set{u,v} \in E}$. 
        These are follower's edges with $c(e) = M$ and $d(e) = 0$, where $M > 3n$ is a large constant.
        Edges from $E'$ are referred to as \emph{shortcuts}.
    \end{itemize}
    The description of the instance $I$ is completed by setting $s := v_1^{(1)}$ as start vertex and $t := v_3^{(n)}$ as end vertex. We will now prove that $\OPTw(I) \leq 3n - k$ if and only if $\alpha(G) \geq k$.

    First, assume that $\alpha(G) \geq k$. Let $W \subseteq V$ be an independent set of size at least $k$ in~$G$.
    Then the leader can choose the set $X := \bigcup_{w \in W} \set{\set{w_1,w_2}, \set{w_2,w_3}} \subseteq \El$ as a leader's solution.
    Note that, since $W$ is an independent set,
    we have for every shortcut $\set{u,v}$ that
    $\set{\set{u_1,u_2}, \set{u_2,u_3}} \in \El \setminus X$ or $\set{\set{v_1,v_2}, \set{v_2,v_3}} \in \El \setminus X$.
    Since the follower is not allowed to use edges in $\El \setminus X$, at least one of the vertices $u_2$ and $v_2$ is a dead end for the follower. Hence, they can never use any shortcut.
    Therefore, the follower goes in order from $v^{(1)}$ to $v^{(n)}$, choosing greedily for every $i \in \fromto{1}{n}$: 
    If the leader's edges that are attached to $v^{(i)}_2$ are in $X$, then the follower uses them.
    The leader pays a cost of $2$ for this pair of edges.
    Otherwise, the follower uses the direct edge $\set{v^{(i)}_1,v^{(i)}_3}$,
    causing additional costs of $3$ for the leader.
    Since $\lvert W \rvert \geq k$, the total leader's costs are at most $3n - k$, which was to show.

    For the opposite direction of the proof, assume that $\OPTw(I) \leq 3n - k$. 
    Due to \cref{lem:weak-optimal-solution-is-path}, there exist $X \subseteq \El$ and $Y \subseteq \Ef$ such that $Y$ is an optimal follower's response to $X$, the set~$X \cup Y$ forms a simple path, and $c(X \cup Y) \leq 3n - k$.
Note that $Y$ does not contain any shortcut edges due to $c(X \cup Y) \leq 3n - k$ and $M > 3n$.
    Now observe that, if there exists some $v \in V$ such that the set $X$ contains exactly one of the two edges $\set{v_1,v_2}$ and $\set{v_2,v_3}$, then the simple path $X \cup Y$ necessarily uses a shortcut edge, a contradiction.
    Hence, for each $v \in V$, the set~$X$ contains either none or both of the two edges $\set{v_1,v_2}$ and $\set{v_2,v_3}$. Let $W \subseteq V$ be those vertices where $X$ contains both. We claim that $W$ is an independent set of size at least $k$. 
    Indeed, suppose that $W$ is not independent, i.e.\ there are $u,v \in W$ with $\set{u,v} \in E$. W.l.o.g.\ we assume that $u$ comes before $v$ in the ordering of $V$ that we chose when constructing $I$. 
    Since $Y$ does not use any shortcuts, the follower goes strictly from left to right in \cref{fig:BSP-weak-NP-complete}, and has costs at least 1 to go from $u_1$ to $v_3$.
    But if the follower used a shortcut, they would have cost at most 0 to go from~
    $u_1$ to $v_3$.
    Hence, $Y$ is not an optimal follower's response to $X$, a contradiction to the choice of $(X, Y)$.
    Finally, in the joint solution~$X \cup Y$, the leader pays a cost of $3$ for every $v \in V \setminus W$ and a cost of $2$ for every $v \in W$. Thus, $c(X \cup Y) \leq 3n - k$ implies $\lvert W \rvert \geq k$. This concludes the proof for the undirected case.

    Finally, for the directed acyclic case, consider \cref{fig:BSP-weak-NP-complete} where every edge is oriented from left to right. The resulting graph is acyclic and it is easily seen that the proof of NP-completeness is analogous.
  \end{proof}

  Observe that one could also apply \cref{lem:undir-special-case-of-dir}
  in order to derive the NP-completeness of {\BSPweakDir} from the NP-completeness of {\BSPweakUndir}.
  However, the construction in the proof of \cref{lem:undir-special-case-of-dir} does not result in an directed acyclic graph,
  in contrast to the explicit construction for {\BSPweakDir} in the proof of \cref{thm:BSP-weak-NP-complete}.

\section{Strong path completion variant}
\label{sec:BSP-strong}
\subsection{Vertex fixing lemma}

In this subsection, we prove a helpful lemma in order to prepare the later hardness results of \cref{sec:BSP-strong,sec:few-leader-edges}. \cref{lem:lifting-lemma} is technical, so we first explain the rough idea behind it: 
Given some instance of {\BSPstrongUndir}, it would be sometimes helpful to be able to modify this instance, by selecting some subset $W \subseteq V$ of the vertices and \enquote{enforcing} that every optimal solution must be a path that includes all the vertices of $W$.  
More precisely, given some arbitrary subset $W \subseteq V$, we consider the following modified bilevel program:
\begin{equation} \label{eq:BSP-modified}
  \begin{aligned}
    \min\ & c(X \cup Y) \\
    \text{s.t. } & X \subseteq \El \\
    & Y \in
    \begin{aligned}[t]
      \argmin\ & d(Y') \\
      \text{s.t. } & Y' \subseteq \Ef \\
      & X \cup Y' \text{ is an $s$-$t$-path whose vertex set includes $W$}
    \end{aligned}
  \end{aligned}
\end{equation}

The following lemma shows that, for every instance of the {\BSP} problem and every vertex set $W$, the original instance can be modified such that the new instance is equivalent to problem \cref{eq:BSP-modified}. However, some special attention is required regarding the case where the newly introduced constraint is infeasible, i.e.\ the case where there is no $s$-$t$-path through the vertex set~$W$.

\begin{lemma}[Vertex fixing lemma]
\label{lem:lifting-lemma}
    Let $I = (G,\El,\Ef, s,t,c,d)$ be an instance of the {\BSP} problem in the strong path completion variant, with an undirected graph $G = (V, E)$.
    Let $W \subseteq V$ be a subset of the vertices.
    Let $\varepsilon > 0$ be a positive constant.
    One can construct a modified instance $I' = (G',E'^\ell,E'^f, s',t',c',d')$ of the same problem and some integer $M > 0$ in polynomial time, such that the following holds:
\begin{enumerate}[(i)]
    \item The new instance $I'$ has optimal value $\OPTs(I') \geq M$ if and only if there is no $s$-$t$-path in the original graph $G$ whose vertex set includes $W$.
    \item In the other case, i.e.\ $\OPTs(I') < M$, we have that the leader's optimal value $\OPTs(I')$ of the new instance is equal to the optimal value of \cref{eq:BSP-modified} for the old instance $I$ and the vertex set $W$. Moreover,
      optimal leader's solutions of the two problems can be constructed from each other in polynomial time and,
      given such corresponding optimal leader's solutions $X'$ and $X$, the follower's optimal values satisfy $\OptFOLs(I', X') = d^\star + \varepsilon$, where $d^\star$ is the follower's optimal value in \cref{eq:BSP-modified} given the instance $I$, the vertex set $W$, and the leader's solution $X$.
    \item The number of leader's edges in the new instance is at most
      $\lvert E'^\ell \rvert \leq 2 \lvert W \rvert + 4 \lvert \El \rvert$.
    \end{enumerate}
  \end{lemma}
\begin{proof}
  We can assume w.l.o.g.\ that $W \subseteq V \setminus \set{s,t}$ because every $s$-$t$-path contains the vertices $s$ and $t$. Let $M := \lceil 1 + \sum_{e \in E} c(e) \rceil$. Note that the integer $M$ is a strict upper bound on the leader's cost of every feasible solution to the old instance $I$.
We now describe the construction of the new instance $I'$ (an accompanying image is given in \cref{fig:lifting-lemma}).

\begin{figure}
    \centering
    \includegraphics[scale=0.9]{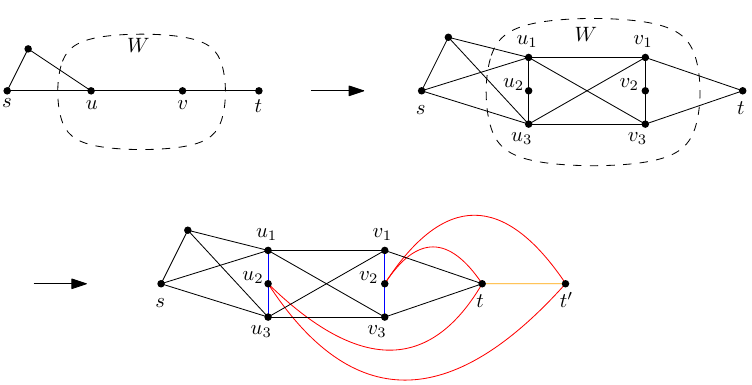}
    \vspace{\baselineskip}
    
    \begin{tabular}{c|ccc} 
        Edge type & Leader/Follower & $c'(e)$ & $d'(e)$ \\
        \hline
        $P_w$, $w \in W$ (\textcolor{blue}{blue}) & L & 0 & 0\\
        $\tilde E$ (\textcolor{red}{red}) & F & M & 0\\
        $\set{t, t'}$ (\textcolor{orange}{orange}) & F & 0 & $\varepsilon$\\
        other (black) & inherited & inherited & inherited
    \end{tabular}
    \caption{Construction used in \cref{lem:lifting-lemma}. The table explains the meaning of the different edge types. Colored figure available online.}
    \label{fig:lifting-lemma}
\end{figure}

Let the graph $G' = (V',E')$ be constructed from the old graph $G = (V, E)$ as follows. 
Every vertex $w \in W$ is removed and replaced by a path $P_w$ consisting of three vertices $w_1,w_2,w_3$ and two edges connecting them in this order.
These two edges are leader's edges (i.e.\ they belong to~$E'^\ell$) and have leader's cost $c'(e) := 0$ and follower's cost $d'(e) := 0$.

If the old graph contained some edge $\set{w, x}$ for $w \in W$ and $x \in V \setminus W$, the new graph contains the two edges $\set{w_1, x}$ and $\set{w_3, x}$. 
These two edges are leader's/follower's edges in~$I'$ if and only if the original edge $\set{w, x}$ was a leader's/follower's edge in $I$, 
and the costs are inherited from the old edge, that is $c'(\set{w_1, x}) = c'(\set{w_3, x}) := c(\set{w, x})$ and likewise $d'(\set{w_1, x}) = d'(\set{w_3, x}) := d(\set{w, x})$.

If the old graph contained some edge $\set{u, v}$, where both $u, v \in W$, then the new graph contains the
four edges $\set{u_1, v_1}$, $\set{u_1, v_3}$, $\set{u_3, v_1}$, and $\set{u_3, v_3}$.
Similarly to the above case, these four new edges are leader's/follower's edges if and only if the original edge $\set{u, v}$ was a leader's/follower's edge. 
Likewise, the leader's and follower's costs are inherited from the original edge.

Every edge $\set{x,y}$ with both $x,y \in V \setminus W$ is not modified, i.e.\ it stays a leader's/follower's edge and keeps its leader's and follower's costs. In a final step, we add one additional vertex $t'$ and the edge set
\begin{align*}
\tilde E := \bigcup_{w \in W} \set{\set{w_2, t}, \set{w_2, t'}}
\end{align*}
to the graph~$G'$. All edges $e \in \tilde E$ are follower's edges and receive leader's cost $c'(e) = M$ and follower's cost $d'(e) = 0$.

We add one more follower's edge $\set{t,t'}$ with leader's cost $c'(e) = 0$ and follower's cost $d'(e) = \varepsilon$.
The description of the instance $I'$ is completed by specifying its source $s' := s$ and sink $t'$.

We quickly explain the main idea behind this construction, before we proceed with the formal proof: The edges in $\tilde E$ are very cheap for the follower, but extremely costly for the leader. Let us call them \emph{dangerous} edges.
Hence the leader must try at all means to stop the follower from using a dangerous edge. We will show that this essentially means that the leader has to include every path $P_w$ in their solution $X'$. 
This forces the follower to walk through every path $P_w$ and makes it impossible to use a dangerous edge. 
It can then be seen that the remaining task is essentially equivalent to solving problem~\cref{eq:BSP-modified}.

We begin with a helpful claim about the new instance $I'$.

\emph{Claim.} If the leader's solution $X'$ does not include both edges of each of the paths $P_w$, i.e.\ if $\bigcup_{w \in W} P_w \subseteq X'$ does not hold, then the follower's problem becomes infeasible or the follower uses a dangerous edge.

\emph{Proof of the claim.} We distinguish the following two cases:

\emph{Case 1}: There exists some $w \in W$ with $\lvert X' \cap P_w \rvert = 0$. If the followers's problem is infeasible, we are immediately done. Otherwise, let $P'$ denote the optimal $s'$-$t'$-path chosen by the follower and consider how $P'$ arrives at the sink $t'$. 
If $P'$ uses a dangerous edge, we are immediately done. 
Otherwise, $P'$ uses the edge $\set{t, t'}$. Now, since both edges of $P_w$ are leader's edges and $X' \cap P_w = \emptyset$ (and since the follower is not allowed to use edges in $E'^\ell \setminus X'$), the path $P'$ does not visit the vertex $w_2$. 
But observe that it would have been cheaper for the follower to use the two dangerous edges $\set{t, w_2}$ and $\set{w_2, t'}$ instead of the edge $\set{t, t'}$ (cost 0 instead of cost $\varepsilon$). This is a contradiction to $P'$ being optimal for the follower and $\set{t, t'}$ being the last edge of $P'$.
We conclude that $P'$ must arrive at $t'$ with a dangerous edge, so we are done.

\emph{Case 2}: There exists some $w \in W$ with $\lvert X' \cap P_w \rvert = 1$. In this case, since the final $s'$-$t'$-path must contain one edge of $P_w$, 
but the follower is not allowed to use the other edge of $P_w$, the problem is either infeasible or the follower necessarily uses a dangerous edge in order to visit~$w_2$ on the final path. This concludes the proof of the helpful claim.

In the following, we will translate $s'$-$t'$-paths in $G'$ that do not use dangerous edges into $s$-$t$-paths in $G$ and vice versa.
In order to formally talk about this, we introduce the following notation.
Given an edge set $F' \subseteq E' \setminus \tilde E$ in $G'$ that does not contain any dangerous edges, let $\tau(F') \subseteq E$ be the result of contracting the paths $P_w$ and the edge~$\set{t, t'}$ in $F'$. In other words, $\tau(F')$ consists of the edges in $G$ from which the edges in $F' \setminus (\bigcup_{w \in W} P_w \cup \set{\set{t, t'}})$ are constructed in the definition of~$I'$.
See \cref{fig:vertex-lemma-path-contract} for an illustration.
\begin{figure}
    \centering
    \includegraphics[scale=1.0]{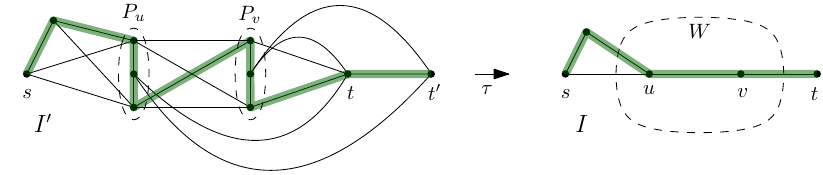}
    \caption{Example of the path contraction used in the proof of \cref{lem:lifting-lemma}.}
    \label{fig:vertex-lemma-path-contract}
\end{figure}
Accordingly, given an edge set $F \subseteq E$, denote by $\tau^{-1}(F) = \set{F' \subseteq E' \setminus \tilde E : \tau(F') = F}$ the preimage of~$F$, which is the set of all edge sets in $G'$ (without dangerous edges) whose contraction results in $F$.
Note that each edge $e \in E$ has either a single corresponding edge in $E'$ or two edges or four edges, which can occur in its preimage~$\tau^{-1}(\set{e})$.
Moreover, the leader's costs satisfy $c'(F') = c(\tau(F'))$ for all $F' \subseteq E' \setminus \tilde E$ that contain at most one of these copies for each~$e \in E$. The follower's costs satisfy $d'(F') = d(\tau(F)) + \varepsilon$ in case $F'$ contains the edge $\set{t, t'}$ or $d'(F') = d(\tau(F))$ otherwise.

We now give the formal proof of the lemma.
Item (iii) is trivial by the definition of the instance~$I'$.
We now consider item (i), which has two directions.

For the first direction, let us assume that there is no $s$-$t$-path in $G$ that includes $W$. 
We have to show that $\OPTs(I') \geq M$.
We claim that, in the new instance $I'$, no matter which solution~$X'$ the leader chooses, the follower's problem is either infeasible or the follower uses a dangerous edge. 
This suffices, since it means that the leader's total cost resulting from $X'$ is either $\infty$ (in the first case) or at least as large as $M$ (in the second case).
Indeed, let $X' \subseteq E'^\ell$ be any solution chosen by the leader. If not $\bigcup_{w \in W} P_w \subseteq X'$, then we are immediately done by using the helpful claim. So let us assume $\bigcup_{w \in W} P_w \subseteq X'$. If the follower's problem is infeasible, we are done. Otherwise, there exists an $s'$-$t'$-path $P'$ in $G'$ that includes $X'$. 
In particular, $P'$ includes all the paths $P_w$ as subpaths. 
Note that $P'$ does not contain any dangerous edge because it is a simple path. Hence, the last edge of $P'$ is the edge $\set{t,t'}$.
However, the contraction $\tau(P')$ of the path $P'$ is an $s$-$t$-path in $G$ that includes~$W$ (compare \cref{fig:vertex-lemma-path-contract}), which is a contradiction.
This finishes the proof of the first direction of item~(i).

In order to show the reverse direction, assume that there is an $s$-$t$-path $P$ in $G$ such that $P$ includes $W$. We have to show that $\OPTs(I') < M$. 
Indeed, it is easy to construct an $s'$-$t'$-path~$P' \in \tau^{-1}(P)$ in $G'$ that includes all $P_w$ as subpaths.
Then $X' := E'^\ell \cap P'$ is a feasible leader's solution. Moreover, given the leader's solution $X'$, the follower cannot use any dangerous edge, since $P_w \subseteq X'$ for all $w \in W$. 
Hence, by the definition of $M$, the leader's objective value resulting from choosing $X'$ is smaller than $M$. This concludes the proof of item~(i).

Finally, we prove item (ii).
Assume that $\OPTs(I') < M$. We have to prove that the new instance $I'$ has the same optimal value as problem~\cref{eq:BSP-modified}. 
Consider some solution $X' \subseteq E'^\ell$ chosen by the leader.
By the helpful claim, if not $\bigcup_{w \in W}P_w \subseteq X'$, then $X'$ is either infeasible or the follower uses a dangerous edge.
In both cases, $X'$ is not an optimal solution, since $\OPTs(I') < M$.
Therefore, every optimal leader's solution $X'$ has the property that $P_w \subseteq X'$ holds for all $w \in W$.
This has two consequences for the follower:
They can never use a dangerous edge, and whenever they encounter some $w_1$ ($w_3$, respectively), they have to immediately go to $w_3$ ($w_1$, respectively) using the path $P_w$.
Applying the map $\tau$ to $s'$-$t'$-paths that the follower can build from such a leader's solution~$X'$ gives corresponding $s$-$t$-paths in $G$ that include~$W$.
More precisely, $\tau$ maps feasible follower's solutions for $\ProblemFOLs(I', X')$ to feasible follower's solutions for problem~\cref{eq:BSP-modified} given the instance~$I$, the vertex set~$W$, and the leader's solution~$\tau(X')$.
On the other hand, any feasible follower's solution~$Y'$ for this instance of problem~\cref{eq:BSP-modified} has at least one corresponding feasible follower's solution for $\ProblemFOLs(I', X')$ in $\tau^{-1}(Y')$.
As observed above, the follower's costs of such corresponding follower's solutions differ by precisely $\varepsilon$, due to the edge~$\set{t, t'}$ being contained in every feasible follower's solution of $\ProblemFOLs(I', X')$.
This implies that also the optimal values of these two follower's problems differ by exactly~$\varepsilon$.
Finally, any feasible leader's solution~$X$ of \cref{eq:BSP-modified} (given the instance $I$ and the vertex set $W$) can be transformed into a feasible leader's solution~$X'$ for the instance $I'$ of our {\BSP} problem by choosing an arbitrary feasible $X' \in \tau^{-1}(X)$ that includes all edges in $\bigcup_{w \in W} P_w$.
Then,
as above, the follower's responses in both problems correspond to each other via $\tau$,
resulting in the same leader's objective value in both cases.
This gives a way to construct optimal leader's solutions for the two problems from each other, as claimed in~(ii),
as well as the fact that both problems have the same leader's optimal value.
This finishes the proof of item~(ii).
\end{proof}

\subsection{NP-completeness of the follower's problem}

In the following, we wish to prove the NP-completeness of the follower's problem in case of strong path completion. 
Note that the follower's problem receives the selected leader's solution~$X$ as part of its input,
but it is not guaranteed that $X$ is actually an optimal choice for the leader.
The strongest mathematical result would be the one proving that the follower's problem is NP-complete, 
even under the additional constraint that $X$ is optimal.
In fact, we prove such a result.
Note however, that such a property of $X$ cannot be tested in polynomial time. 
In order to avoid technical complications with the definition of NP-problems, we choose the following formulation of this fact.

\begin{theorem} \label{thm:BSP-strong-follower-NP-complete}
    For the strong path completion variant of the {\BSP} problem (both {\BSPstrongUndir} and {\BSPstrongDir}), the follower's problem is NP-complete. More precisely, there exists a polynomial-time reduction from the {\HamPath} problem to the follower's problem $\ProblemFOLs(I, X)$ with the additional property that $X$ is an optimal leader's solution.
\end{theorem}
\begin{proof}
  The directed case follows from the undirected case by \cref{lem:undir-special-case-of-dir}.
  For the undirected case, we present a reduction from the NP-complete {\HamPath} problem~\cite{garey1979computers}.
  Given a graph $G = (V, E)$ and vertices $s, t \in V$,
  the question is whether there is an $s$-$t$-path in $G$ that visits every vertex exactly once. 

  Let $(G, s, t)$ be an instance of the {\HamPath} problem with $G = (V, E)$.
  We construct an instance $I = (G', \El, \Ef, s, t, c, d)$ of {\BSPstrongUndir} as follows.
  Let $G'$ be the complete graph on the vertex set~$V$ of $G$, i.e.\ $G'$ contains every possible edge. 
  The start vertex $s$ and the end vertex $t$ in $I$ are the same as in the {\HamPath} instance.
  Let $\El = \emptyset$ and $\Ef = E(G')$.
  We define the leader's costs as $c(e) = 0$ for all $e \in E(G')$
  and the follower's costs as $d(e) = 0$ for all $e \in E(G) \cap E(G')$ and $d(e) = 1$ for all $e \in E(G') \setminus E(G)$.
  In this instance,
  the leader has nothing to decide
  and the follower looks for an $s$-$t$-path in $G'$ that uses a minimal number of edges that are not in $G$.
  
  We now apply the vertex fixing lemma (\cref{lem:lifting-lemma}) to the instance~$I$ and the vertex set $W = V \setminus \set{s, t}$.
  We obtain a modified instance $I'$.
  Note that, since the complete graph~$G'$ contains a Hamiltonian $s$-$t$-path, we are in the second case of the vertex fixing lemma.
  In the corresponding problem~\cref{eq:BSP-modified},
  the leader still has nothing to decide
  and the follower's task is to find a Hamiltonian $s$-$t$-path in $G'$ that uses a minimal number of edges that are not in $G$.
  In particular, if $G$ contains a Hamiltonian $s$-$t$-path, then the follower's objective value is 0, otherwise it is at least 1.
  Let $X'$ be the unique optimal leader's solution for instance $I'$,
  as given by the vertex fixing lemma (corresponding to the unique leader's solution $X = \emptyset$ in \cref{eq:BSP-modified}).
  The lemma implies that, if $G$ contains a Hamiltonian $s$-$t$-path, then the follower's optimal value $\OptFOLs(I', X')$ is~$\varepsilon$, otherwise it is at least~$1 + \epsilon$. (Note that $\varepsilon > 0$ can be set arbitrarily in \cref{lem:lifting-lemma}.)
  Thus, we have $\OptFOLs(I', X') \leq \varepsilon$ if and only if $G$ contains a Hamiltonian $s$-$t$-path.
\end{proof}

\subsection[Sigmap2-completeness of the leader's problem]{$\Sigma^p_2$-completeness of the leader's problem}
\label{subsec:sigma-2-completeness}

In this section, we prove that the {\BSP} problem in the strong path completion variant is even harder\footnote{assuming $\mathrm{NP} \neq \Sigma^p_2$} than NP-complete, namely $\Sigma^p_2$-complete. 

\begin{theorem}
\label{thm:BSP-strong-sigma-2-complete}
    Both problems {\BSPstrongUndir} and {\BSPstrongDir} are $\Sigma^p_2$-complete.
\end{theorem}

The proof of \cref{thm:BSP-strong-sigma-2-complete} consists of two ingredients. The first ingredient is the vertex fixing lemma (\cref{lem:lifting-lemma}), which roughly states that we can always modify an instance such that certain vertices have to be visited by the path. 
The second ingredient is inspired by a recent line of research showing how to obtain $\Sigma^p_2$-hardness results for many min-max optimization problems \cite{GrueneWulf2024, GoerigkLW24}, see also \cite{jackiewicz2024computational}. 
The main idea is to take some existing NP-complete problem, say, the {\HamPath} problem, examine its NP-completeness proof closely, and upgrade this existing NP-completeness 
proof to a $\Sigma^p_2$-completeness proof of some min-max variant of the {\HamPath} problem.
This is useful if the min-max variant of {\HamPath} is similar to the target problem (the {\BSP} problem)
because $\Sigma^p_2$-hardness of the target problem can then be shown by a relatively easy reduction from the min-max variant of {\HamPath}. 
Even though the setting of \cite{GrueneWulf2024, GoerigkLW24} is not immediately applicable to our bilevel setting, we show that the same idea still works. However, we require several nontrivial adaptations specific to the bilevel setting.
Formally, we consider the following min-max variant of the {\HamPath} problem.
In the following, let $\delta(v)$ denote the set of edges incident to some vertex $v$.

\begin{quote}
    \textbf{Problem} {\MinMaxHam}

    \textbf{Input:} An undirected graph $G = (V, E)$, vertices $s,t \in V$, a vertex $v \in V$ of degree~$3$, an edge $\tilde e \in \delta(v)$ incident to $v$, a subset $B \subseteq E$ of the edges with $B \cap \delta(v) = \emptyset$.
    
    \textbf{Question:} Is there a subset $B' \subseteq B$ with the following properties: There is at least one Hamiltonian $s$-$t$-path $H$ in $G$ with $H \cap B = B'$, and, for all Hamiltonian $s$-$t$-paths~$H$ in $G$, we have the implication $(H \cap B = B' \implies \tilde e \not\in H)$?
\end{quote}

The problem {\MinMaxHam} can be quickly summarized as follows: Is there a set $B' \subseteq B$ such that, by fixing $B'$ as part of a Hamiltonian path, we enforce that such a path never uses the edge $\tilde e$ (and at least one such path exists)?
Note that we make two technical assumptions in the problem definition, namely that $v$ has degree 3 and that $B \cap \delta(v) = \emptyset$ holds. These will be convenient for our proof of \cref{thm:BSP-strong-sigma-2-complete}, and they are not a restriction because they are naturally satisfied for the instance constructed in the hardness proof of \cref{lem:bilevel-ham-path} below.

The remainder of this subsection is structured as follows: First, we show that {\MinMaxHam} is $\Sigma^p_2$-complete. 
Then we show that {\MinMaxHam} can be reduced to the {\BSP} problem,
by using the vertex fixing lemma. This suffices to prove \cref{thm:BSP-strong-sigma-2-complete}.

In order to prove the $\Sigma^p_2$-completeness of {\MinMaxHam}, we reduce from the canonical problem {\QSAT}.
We introduce some notation. 
For a set $X = \fromto{x_1}{x_n}$ of Boolean variables, the corresponding set of \emph{literals} is $L_X = \fromto{x_1}{x_n} \cup \fromto{\overline x_1}{\overline x_n}$. 
A \emph{clause} is a disjunction of literals. 
A Boolean formula is in \emph{conjunctive normal form} (CNF) if it is a conjunction of clauses. It is in 3-CNF if additionally each clause consists of at most three literals.
It is in \emph{disjunctive normal form} if it is a disjunction of conjunctions of literals. It is in 3-DNF if each conjunction has at most three literals.
An \emph{assignment} of the variables is a map $\alpha : X \to \set{0,1}$. Let $\varphi(X)$ be a Boolean formula depending on the variables $X$. 
We denote by $\varphi(\alpha) \in \set{0,1}$ the evaluation of $\varphi$ under the assignment $\alpha$. For formulas whose variable set is partitioned into two disjoint parts $X$ and $Y$, we extend the above notation to $\varphi(X, Y)$ and $\varphi(\alpha, \beta)$. 
In the well-known NP-complete {\SAT} problem, we are given a formula in 3-CNF, and the question is if there is a satisfying assignment.
The canonical $\Sigma^p_2$-complete problem {\QSAT} \cite{stockmeyer1976polynomial,wrathall1976} is then defined as follows. (Note that, in contrast to {\SAT}, the input here is in DNF.)

\begin{quote}
    \textbf{Problem} {\QSAT}

    \textbf{Input:} A Boolean formula $\varphi(X, Y)$ in 3-DNF on the variable set $X \cup Y$, where $X$~and~$Y$ are disjoint.
    
    \textbf{Question:} Does there exist an assignment $\alpha : X \to \set{0,1}$ such that, for all assignments $\beta : Y \to \set{0,1}$, we have that $\varphi(\alpha, \beta) = 1$?
  \end{quote}

As stated above, our proof works by extending some existing NP-completeness proof of the {\HamPath} problem to the min-max case.
It turns out that it is not necessary to repeat all technical details of such a proof. 
All that we require is that there exists some reduction with the key properties specified in \cref{obs:ham-cycle-NP-reduction}. 
It can be easily verified that, for example, the reduction of Garey, Johnson, and Tarjan \cite{garey1976planar} has this property. (Note that theirs is a reduction to the Hamiltonian cycle problem, but can be easily modified to the {\HamPath} problem by standard arguments.)
Many other folklore reductions from {\SAT} to the {\HamPath} problem also have this property.

\begin{observation}
    \label{obs:ham-cycle-NP-reduction}
    There exists a polynomial-time reduction with the following properties: Given as input a Boolean formula $\varphi$ in CNF on the variable set $X := \set{x_1,\dots,x_n}$, it computes an undirected graph $G_\varphi = (V, E)$, together with vertices $s,t \in V$ and $2n$ distinct edges $F := \fromto{e_{x_1}}{e_{x_n}} \cup \fromto{e_{\overline x_1}}{e_{\overline x_n}} \subseteq E$ such that the following properties are satisfied:
\end{observation}
\begin{enumerate}[(i)]
    \item For every $i \in \fromto{1}{n}$, every Hamiltonian $s$-$t$-path $H$ in $G_\varphi$ uses exactly one of the two edges $e_{x_i}$ and $e_{\overline x_i}$.
    \item There is a direct correspondence between subsets of the edge set $F$ that can be part of a Hamiltonian $s$-$t$-path and satisfying assignments for the formula $\varphi$.
    Formally: For an assignment $\alpha : X \to \set{0,1}$, let $F_\alpha$ denote the corresponding set of edges defined by $F_\alpha = \set{e_{x_i} : \alpha(x_i) = 1} \cup \set{e_{\overline x_i} : \alpha(x_i) = 0}$. Given an assignment $\alpha$, there exists a Hamiltonian $s$-$t$-path $H$ in $G$ with $H \cap F = F_\alpha$ if and only if $\varphi(\alpha) = 1$.
    \item The edges $e_{x_n}$ and $e_{\overline x_n}$ meet at the same vertex $v \in V$, and $v$ has degree 3 in $G_\varphi$, where the third edge incident to $v$ is not contained in~$F$.
\end{enumerate}

\begin{lemma}
\label{lem:bilevel-ham-path}
    The problem {\MinMaxHam} is $\Sigma^p_2$-complete.
\end{lemma}
\begin{proof}
    We have to show containment in $\Sigma^p_2$ as well as hardness. For the containment, consider an instance $(G, s, t, v, \tilde e, B)$ of {\MinMaxHam}, and let $\mathcal{H}$ denote the set of all Hamiltonian $s$-$t$-paths in $G$. Then the problem can be rewritten in the following way:
\begin{align*}
    \exists B' \subseteq B \ \ \exists H_1 \subseteq E \ \ \forall H_2 \subseteq E : \ &(H_1 \in \mathcal{H}) \land (H_1 \cap B =  B') \\
    &\land [((H_2 \in \mathcal{H})  \land (H_2 \cap B = B')) \implies \tilde e \not \in H_2]
\end{align*}
Since the statement on the right can be checked in polynomial time given $B', H_1, H_2$, we conclude that  {\MinMaxHam} is contained in the class $\Sigma^p_2$.

For the hardness proof, we need some preparation. Let an instance $\exists X \forall Y \varphi(X,Y)$ of {\QSAT} be given, 
where $\varphi$ is a Boolean formula in 3-DNF and $X = \fromto{x_1}{x_n}$ and $Y = \fromto{y_1}{y_n}$ are its variables. 
Note that we can w.l.o.g.\ assume that $\lvert X \rvert = \lvert Y \rvert$. 
We consider a sequence of transformations of $\varphi =: \varphi_0$. 
First, let 
\begin{align*}
    \varphi_1 := \neg \varphi.
\end{align*}
Note that $\varphi_1$ is in 3-CNF by De Morgan's laws. 
Hence, we can write it as $\varphi_1 = (p_{11} \lor p_{12} \lor p_{13}) \land \dots \land (p_{k1} \lor p_{12} \lor p_{k3})$,
where $p_{ij}$ is a literal over $X \cup Y$ for $i \in \fromto{1}{k}$ and $j \in \set{1,2,3}$. 
We introduce a new variable $z$. 
Our goal is to define a formula $\varphi_2$ with the intended meaning \enquote{$z \Leftrightarrow \varphi_1(X, Y)$}, i.e.\ the variable $z$ encodes the truth value of $\varphi_1$.
The definition of $\varphi_2$ in 3-CNF can be obtained in the following standard way:
We introduce $3k$ new auxiliary variables $A = \fromto{a_1}{a_k} \cup \fromto{a'_1}{a'_k} \cup \fromto{a''_1}{a''_k}$.
The variable $a_i$ has the intended meaning \enquote{clause~$i$ of $\varphi_1$ is satisfied}, and the variable $a'_i$ has the intended meaning \enquote{all the clauses $1, \dots, i$ of $\varphi_1$ are satisfied}.
(The meaning of the variable $a''_i$ is explained below.)
Now $\varphi_2$ is obtained from $\varphi_1$ in the following way:
\begin{itemize}
\item For each $i \in \fromto{1}{k}$, we add the five clauses $(\neg a_i \lor p_{i1} \lor a''_i), (\neg a''_i \lor p_{i2} \lor p_{i3}), (\neg p_{i1} \lor a_i), (\neg p_{i2} \lor a_i), (\neg p_{i3} \lor a_i)$.
  Note that it is possible to satisfy these five clauses with an appropriate choice of $a''_i$ if and only if the remaining variables meet the following condition:
    \begin{align*}
        a_i \Leftrightarrow ( p_{i1} \lor p_{i2} \lor p_{i3})
    \end{align*}
    \item For each $i \in \fromto{2}{k}$, we add the three clauses $(\neg a_i \lor \neg a'_{i-1} \lor a'_i), (\neg  a'_i \lor a_i), (\neg a'_i \lor a'_{i-1})$. Note that these three clauses are equivalent to the formula 
    \begin{align*}
    a'_i \Leftrightarrow (a_i \land a'_{i-1}).
    \end{align*}
    \item We add the two clauses $(z \lor \neg a'_k), (\neg z \lor a'_k)$. Note that these are equivalent to the formula $z \Leftrightarrow a'_k$.
\end{itemize}
This completes our description of the formula $\varphi_2$. Note that $\varphi_2 = \varphi_2(X,Y,A,z)$ is in 3-CNF and depends on the variables $X,Y,A,z$. 
Furthermore, it follows from the arguments above that an assignment of these variables satisfies $\varphi_2$ if and only if the assignment of $z$ is the same value as the evaluation of $\varphi_1(X,Y)$, and the auxiliary variables $A$ are set accordingly.

It remains to show hardness of {\MinMaxHam}. So let an instance $\exists X \forall Y \varphi(X,Y)$ of {\QSAT} be given. 
We define an instance $(G, s, t, v, \tilde e, B)$ of {\MinMaxHam} in the following way: We let $\varphi_2$ be the formula derived from $\varphi$ as described above.
For the graph $G$ and the vertices $s$ and $t$, we consider the transformation from \cref{obs:ham-cycle-NP-reduction} applied to the formula $\varphi_2$, that is $G := G_{\varphi_2}$.
The set $B$ is defined by $B := \bigcup_{x \in X} \set{e_x, e_{\overline x}}$, where $X = \fromto{x_1}{x_n}$ and where $e_{x_i}$ and $e_{\overline x_i}$ are defined like in \cref{obs:ham-cycle-NP-reduction}.
Note that $B$ contains the edges corresponding to variables in $X$, but not those edges corresponding to variables in $Y \cup A \cup \set{z}$.
We let vertex $v$ be the common endpoint of the edges $e_z$ and $e_{\overline z}$, which exists and has degree 3 by \cref{obs:ham-cycle-NP-reduction}~(iii).
Note that \cref{obs:ham-cycle-NP-reduction}~(iii) also implies that $B \cap \delta(v) = \emptyset$ is satisfied.
Finally, we let $\tilde e := e_z$. This completes the description of the instance $I = (G, s, t, v, \tilde e, B)$.

We claim that $I$ is a yes-instance of {\MinMaxHam} if and only if $\varphi$ is a yes-instance of {\QSAT}.

Indeed, for the \enquote{if} direction, assume that $\varphi$ is a yes-instance of {\QSAT}.
Hence, there is an assignment
$\alpha : X \to \set{0,1}$ such that for all assignments $\beta : Y \to \set{0,1}$ we have $\varphi(\alpha, \beta) = 1$. 
Let $B_\alpha \subseteq B$ be the subset of $B$ which corresponds to the assignment $\alpha$
(see \cref{obs:ham-cycle-NP-reduction}~(ii)), and consider $B' = B_\alpha$. 
We have to show that there exists at least one Hamiltonian $s$-$t$-path~$H$ in $G$ with $H \cap B = B_\alpha$ and that no such Hamiltonian path contains the edge $e_z$. 
First, we prove that at least one such path $H$ exists. 
Choose an arbitrary assignment $\beta : Y \to \set{0,1}$, 
and then choose an assignment of the auxiliary variables $A$ and of the variable $z$ accordingly so that all clauses of~$\varphi_2$ are trivially satisfied. This extends the partial assignment $\alpha$ to a satisfying assignment of~$\varphi_2$.
By \cref{obs:ham-cycle-NP-reduction} (ii), we conclude that there exists a Hamiltonian $s$-$t$-path~$H$ with $H \cap B = B_\alpha$.
Secondly, for every assignment $\beta : Y \to \set{0,1}$, we have that $\varphi(\alpha, \beta) = 1$, by the choice of $\alpha$. 
Every Hamiltonian $s$-$t$-path $H$ with $H \cap B = B_\alpha$ \enquote{satisfies} the formula~$\varphi_2$ (again by \cref{obs:ham-cycle-NP-reduction}), and $\varphi_2$ is equivalent to $z \Leftrightarrow \varphi_1(X,Y)$, which is equivalent to $z \Leftrightarrow \neg \varphi(X,Y)$. 
Therefore, we have that any Hamiltonian $s$-$t$-path $H$ with $H \cap B = B_\alpha$ always uses the edge $e_{\overline z}$ and never uses the edge $e_z$. In conclusion, we have shown that $I$ is a yes-instance of {\MinMaxHam}.

For the \enquote{only if} direction, assume that $\varphi$ is a no-instance. 
Hence, for all assignments $\alpha : X \to \set{0,1}$, there exists an assignment $\beta : Y \to \set{0,1}$ with $\varphi(\alpha, \beta) = 0$.
We claim that no set $B' \subseteq B$ has the properties desired in the definition of {\MinMaxHam}, and hence $I$ is a no-instance.
Indeed, recall that $B = \bigcup_{x \in X} \set{e_x, e_{\overline x}}$. Now, if $B'$ contains both edges $e_{x_i}$ and $e_{\overline x_i}$ for some $i \in \fromto{1}{n}$ or none of the two, then, by \cref{obs:ham-cycle-NP-reduction} (i), the set $B'$ cannot be extended to a Hamiltonian $s$-$t$-path in $G$. 
Hence, $B'$ does not meet the first requirement for a yes-solution of {\MinMaxHam}.
On the other hand, if $B'$ contains exactly one of the edges $e_{x_i}$ and $e_{\overline x_i}$ for all $i \in \fromto{1}{n}$, then $B'$ naturally encodes some assignment $\alpha : X \to \set{0,1}$.
By assumption, there is some assignment $\beta : Y \to \set{0,1}$ with $\varphi(\alpha, \beta) = 0$ and hence $\varphi_1(\alpha, \beta) = 1$. 
Since $\varphi_2$ is equivalent to $z \Leftrightarrow \varphi_1(X,Y)$, there is a way to extend the assignment $\alpha$ to a satisfying assignment of $\varphi_2$ such that variable $z$ is assigned \enquote{1}. By \cref{obs:ham-cycle-NP-reduction}~(ii), this shows that there is a Hamiltonian $s$-$t$-path $H$ with $B \cap H = B'$ and $e_z = \tilde e \in H$.
Hence, $B'$ does not meet the second requirement for a yes-solution of {\MinMaxHam}. This completes the argument and we have shown that $I$ is a no-instance of {\MinMaxHam}.

In conclusion, we have shown that $\varphi$ is a yes-instance of {\QSAT} if and only if $I$ is a yes-instance of {\MinMaxHam}. Thus, {\MinMaxHam} is $\Sigma^p_2$-complete.
\end{proof}

We can now combine our tools and prove that the {\BSP} problem is $\Sigma^p_2$-complete.

\begin{proof}[Proof of \cref{thm:BSP-strong-sigma-2-complete}]
    We have to show containment in $\Sigma^p_2$ as well as hardness. 
    For the containment, both in the directed and in the undirected case, consider an instance $I = (G, \El, \Ef, s,t,c,d)$ of one of the decision problems {\BSPstrongDir} and {\BSPstrongUndir}. 
    The question is, for a given threshold $T \in \R_{\geq 0}$, whether $\OPTs(I) \leq T$.
    This question can be rewritten as
    \begin{align*}
        \exists X \subseteq \El \ \ \exists Y \subseteq \Ef \ \ \forall Y' \subseteq \Ef : \ &(c(X \cup Y) \leq T) \land (X \cup Y \in \Pst) \\
        &\land \ [(X \cup Y' \in \Pst) \implies d(Y) \leq d(Y')].
    \end{align*}
    We claim that this expression is true if and only if $\OPTs(I) \leq T$, where the optimum is defined like in \cref{eq:BSP-strong}. 
    Indeed, this expression is true if and only if, for some leader's choice $X$,
    there exists at least one follower's solution~$Y$ that is feasible (i.e.\ $X \cup Y \in \Pst$) and also optimal for the follower's problem (i.e.\ it minimizes $d(Y)$)
    and satisfies $c(X \cup Y) \leq T$.
    Since we consider the optimistic setting,
    this is equivalent to $\OPTs(I) \leq T$. 
    This shows that {\MinMaxHam} is contained in the class $\Sigma^p_2$. 

    For the hardness, it suffices to show that the undirected case {\BSPstrongUndir} is $\Sigma^p_2$-hard. Due to \cref{lem:undir-special-case-of-dir}, this also proves that {\BSPstrongDir} is $\Sigma^p_2$-hard.
    We reduce from the $\Sigma^p_2$-complete problem {\MinMaxHam} (see \cref{lem:bilevel-ham-path}).

    Given an instance $I = (G, s,t, v, \tilde e, B)$ of {\MinMaxHam}, we create an instance $I'$ of {\BSPstrongUndir} in the following way:
    The graph $G = (V,E)$ as well as $s$ and $t$ remain the same. We let $\El := B$ and $\Ef := E \setminus B$. All edges $e \in E \setminus \delta(v)$ that are not incident to $v$ have leader's and follower's cost $c(e) = d(e) = 0$. 
    The three edges~$e_1$, $e_2$, and $\tilde e$ that are incident to $v$ have costs $c(e_1) = c(e_2) = 0$, $d(e_1) = d(e_2) = 1$, $c(\tilde e) = 1$, and $d(\tilde e) = 0$.
    Note that these three edges are follower's edges because $B \cap \delta(v) = \emptyset$ by the definition of {\MinMaxHam}.
    This completes the description of instance $I'$.
    We now create a new instance~$I''$ out of the instance $I'$ by applying the vertex fixing lemma (\cref{lem:lifting-lemma}) to instance $I'$ and the vertex set $W:= V$.
    We claim that the new instance $I''$ satisfies $\OPTs(I'') = 0$ if and only if $I$~is a yes-instance of {\MinMaxHam} (and $\OPTs(I'') \geq 1$ otherwise). 

    For the \enquote{if} direction, assume that $I$ is a yes-instance of {\MinMaxHam}. Then, in particular, there exists a Hamiltonian $s$-$t$-path in $G$. 
    Hence, case (i) of \cref{lem:lifting-lemma} does not apply and we must be in case (ii).  
    By the lemma, this means that $\OPTs(I'')$ is equal to the optimal value of the instance $I'$ in the modified version of {\BSPstrongUndir} where the leader's solution~$X$ and the follower's solution $Y$ have to form a Hamiltonian $s$-$t$-path (instead of any $s$-$t$-path) in~$G$.
    Since all three edges $e_1$, $e_2$, and $\tilde e$ incident to $v$ are follower's edges, 
    the follower therefore uses exactly two of these three edges, in any feasible solution.
    Since the follower's costs are equal to 0 on all other edges, the follower has total costs 2 if and only if they use both of the edges $e_1$ and~$e_2$, and total costs 1 if and only if they use the edge $\tilde e$. 
    In the first case, the leader has total costs 0, and in the second case, the leader has total costs 1.
    Since $I$ is a yes-instance of {\MinMaxHam}, there exists a set $B' \subseteq B = \El$ such that every Hamiltonian $s$-$t$-path $H$ with $H \cap \El = B'$ satisfies $\tilde e \not\in H$ and there exists at least one such a path. 
    This implies that, if the leader selects the solution $X = B'$, then the follower's problem is feasible and the leader's total costs are 0. Hence, $\OPTs(I'') = 0$.

    For the \enquote{only if} direction, assume that $\OPTs(I'') = 0$. This means that we are in case (ii) of \cref{lem:lifting-lemma} (since $M > 0$, where $M$ is defined as in the lemma). 
    Hence, $\OPTs(I'') = 0$ is equal to the optimal value of the instance $I'$ in the modified version of {\BSPstrongUndir} where the leader's solution~$X$ and the follower's solution $Y$ have to form a Hamiltonian $s$-$t$-path in~$G$.
    Therefore, there must exist a subset $X \subseteq \El$ such that the follower's problem is feasible, but the follower cannot complete $X$ to a Hamiltonian $s$-$t$-path $H$ that uses the edge $\tilde e$ (since such a path is always the cheapest for the follower, but always has nonzero leader's costs).
    Hence, by setting $B' := X$, we see that $I$ is a yes-instance of {\MinMaxHam}.
    This shows that $I$ is a yes-instance of {\MinMaxHam} if and only if $\OPTs(I'') = 0$, and completes the proof.
  \end{proof}

\subsection{Polynomial-time algorithms for directed acyclic graphs}
While the previous subsections contained many hardness results, we complement them with a positive result in this subsection.
We show that the {\BSP} problem with strong path completion can be solved in polynomial time in the special case of directed acyclic graphs.
This fact may be considered surprising because, by \cref{thm:BSP-strong-sigma-2-complete}, the same problem is $\Sigma^p_2$-complete on general directed graphs. 
Hence, the restriction to acyclic graphs causes a significant drop in complexity, by two stages in the polynomial hierarchy. 
Furthermore, in case of directed acyclic graphs, the weak path completion variant of the problem is NP-complete by \cref{thm:BSP-weak-NP-complete}. Hence, in this special case, the strong path completion variant is easier than the weak path completion variant, in contrast to the general setting.
The reason for this behavior is essentially that, in this special case, the space of solutions to the {\BSP} problem is well-structured and this can be exploited by using a dynamic programming approach.

In order to explain our solution strategy, we introduce some simple concepts.
Given an directed acyclic graph $G = (V,E)$ with $E = \El \cup \Ef$, and two distinct edges $e,e' \in E$, we write $e \prec e'$ if there is a directed path $P \subseteq \El \cup \Ef$ whose first edge is $e$ and whose last edge is $e'$.
It is easy to see that, for any two distinct edges $e,e' \in E$, exactly one of the following three is true: 
Either $e \prec e'$ or $e' \prec e$, or there is no directed path in $G$ from the head of $e$ to the tail of $e'$ or vice versa.
We denote the latter case by $e \parallel e'$. Since $G$ is acyclic, $e_1 \prec e_2$ and $e_2 \prec e_3$ imply $e_1 \prec e_3$.
Given a set $E' \subseteq E$, 
we say that $E'$ is a \emph{chain} if, for $k := \lvert E' \rvert$, the elements of $E'$ can be ordered such that $E'  = \fromto{e_1}{e_k}$ and $e_1 \prec e_2 \prec \dots \prec e_k$.
If $E'$ is not a chain, then there are $e, e' \in E'$ with $e \parallel e'$.
For two distinct edges $e = (u,v), e' = (u', v') \in E$ with $e \prec e'$, let 
$\mathcal{P}^f_{ee'} := \set{P : P \subseteq \Ef \text{ is a $v$-$u'$-path using only follower's edges}}$. 
Note that it is possible that $\mathcal{P}^f_{ee'} = \emptyset$ even though $e \prec e'$, since the definition of the relation~$\prec$ allows to use leader's and follower's edges on the path, but in $\mathcal{P}^f_{ee'}$, we are only allowed to use follower's edges.

For the remainder of this subsection, we make the following assumption about the given input instance $(G, \El, \Ef, s,t, c, d)$: 
We assume that the source vertex $s$ has only a single outgoing edge (which we denote by $e_s$) and the sink vertex $t$ has only a single incoming edge (which we denote by $e_t$). Hence, every $s$-$t$-path starts with edge $e_s$ and ends with edge $e_t$. Furthermore, we assume that $e_s,e_t \in \El$ and $c(e_s) = d(e_s) = c(e_t) = d(e_t) = 0$.
These assumptions will simplify the notation. Note that they can be made w.l.o.g., since a new source vertex and a new sink vertex together with the edges $e_s$ and $e_t$ can be added artificially if necessary. 

\begin{lemma}
\label{lem:BSP-strong-follower-acyclic}
    The follower's problem of {\BSPstrongDir} can be solved in polynomial time on directed acyclic graphs.
\end{lemma}
\begin{proof}
  Consider an instance $I = (G, \El, \Ef, s, t, c, d)$ of the {\BSP} problem with an directed acyclic graph $G = (V,E)$,
  and let a leader's solution $X \subseteq \El$ be given.  
  We distinguish two cases.
  If $X$ is not a chain, then there are distinct edges $e, e' \in X$ with $e \parallel e'$.
  However, in this case, the follower's problem is infeasible, since the strong path completion forces the follower to use both $e$ and $e'$ on a path, which is impossible.
  If $X$ is a chain, then we have $e_1 \prec e_2 \prec \dots \prec e_k$ for $X = \fromto{e_1}{e_k}$, where $k := \lvert X \rvert$.

  Observe that one can detect whether $X$ is a chain, and simultaneously determine the ordering $e_1, \dots, e_k$ in case $X$ is a chain, as follows:
  Compute a topological vertex ordering of the directed acyclic graph $G$.
  Let $e_1 = (u_1, v_1), \dots, e_k = (u_k, v_k)$ be the ordering of $X$ for which $u_1, \dots, u_k$ appear in the topological ordering in this order.
  Now, for every $i \in \fromto{1}{k-1}$, check whether there is a directed path from $v_i$ to $u_{i+1}$ in $G$.
  If this is the case for all $i \in \fromto{1}{k-1}$, then the order $e_1, \dots, e_k$ satisfies $e_1 \prec \dots \prec e_k$.
  Otherwise, for some $i \in \fromto{1}{k-1}$, there is neither a directed path from $v_i$ to the $u_{i+1}$ nor a directed path from $v_{i+1}$ to $u_i$ (because of the topological ordering). Therefore, $e_i \parallel e_{i+1}$ and $X$ is not a chain.

  From now on, assume that $X$ is a chain and that the corresponding ordering $e_1, \dots, e_k$ is known.
  Note that $X$ can only be feasible if $e_1 = e_s$ and $e_k = e_t$, by our assumption about the input instance.
  The follower has to traverse the leader's edges $e_1, \dots, e_k$ in this order.
  Between two leader's edges $e_i$ and $e_{i+1}$, the follower is only allowed to use follower's edges, i.e.\ they need to use a path from $\mathcal{P}^f_{e_ie_{i+1}}$. 
  If, for some $i \in \fromto{1}{k-1}$, we have $\mathcal{P}^f_{e_ie_{i+1}} = \emptyset$, then the follower's problem is infeasible. 
    
    In the case where the follower's problem is feasible, define $\alpha_i := \min \set{d(P) : P \in \mathcal{P}^f_{e_i e_{i+1}}}$ for all $i \in \fromto{1}{k-1}$. 
    It is easily seen that the value of the follower's problem is given by
    $\OptFOLs(I, X) = \sum_{i=1}^{k-1} \alpha_i$.
    (Note that $d(e_i) = 0$ for all $i \in \fromto{1}{k}$, as $e_i \in \El$.)
    Using a standard shortest path algorithm, 
    it can be decided in polynomial time whether $\mathcal{P}^f_{e_ie_{i+1}} = \emptyset$ and the values $\alpha_i$ (and the corresponding paths) can be computed.
    
    Finally note that, if there are multiple optimal paths in $\mathcal{P}^f_{e_i e_{i+1}}$ with the same costs $\alpha_i$, then, due to the optimistic assumption, 
    the follower chooses among them a path minimizing $c(P)$. 
    Such a path can be determined by finding a path which minimizes the objective $d'(P) = d(P) + \varepsilon c(P)$ for some small enough $\varepsilon$, for example $\varepsilon < (\min_{e \in E, d(e) \neq 0} d(e)) / (1 + \sum_{e \in E}c(e))$.
\end{proof}

\begin{theorem} \label{thm:BSP-strong-acyclic-P}
    The problem {\BSPstrongDir} can be solved in polynomial time on directed acyclic graphs.
\end{theorem}
\begin{proof}
    Consider an instance $I = (G, \El, \Ef, s, t, c, d)$ of the {\BSP} problem with an directed acyclic graph $G = (V,E)$.
    Given two edges $e$ and $e'$ with $e \prec e'$, let $\alpha(e,e') := \min\set{d(P) : P \in \mathcal{P}^f_{ee'}}$ if $\mathcal{P}^f_{ee'} \neq \emptyset$ and 
    $\alpha(e,e') := \infty$ otherwise.
    Furthermore, let $\beta(e, e') := \min\set{c(P) : P \in \mathcal{P}^f_{ee'}, d(P) = \alpha(e,e')}$ if $\mathcal{P}^f_{ee'} \neq \emptyset$ and 
    $\beta(e,e') := \infty$ otherwise.
    These numbers can be interpreted as follows: The value $\alpha(e,e')$ is the follower's cost for going from $e$ to $e'$ (using only follower's edges) and $\beta(e,e')$ is the corresponding cost that is caused for the leader by this follower's subpath.
    If it is impossible to go from $e$ to $e'$ using only follower's edges, i.e.\ if $\mathcal{P}^f_{ee'} = \emptyset$, then both numbers are infinite.
    Note that both $\alpha(e,e')$ and $\beta(e,e')$ can be computed in polynomial time, for example by minimizing the expression $d(P) + \varepsilon c(P)$ over all paths $P \in \mathcal{P}^f_{e e'}$.
    
    Like in the proof of \cref{lem:BSP-strong-follower-acyclic}, feasible leader's solutions are chains $X = \fromto{e_1}{e_k}$ with $e_s = e_1 \prec e_2 \prec \dots \prec e_k = e_t$ and $\mathcal{P}^f_{e_i e_{i+1}} \neq \emptyset$ for all $i \in \fromto{1}{k-1}$, and optimal follower's solutions are given by connecting the edges in $X$ by follower's subpaths from the sets $\mathcal{P}^f_{e_i e_{i+1}}$ with costs $\alpha(e_i, e_{i+1})$.
    Accordingly, the leader's optimization problem can be written as
    \begin{align}
      \label{eq:BSP-strong-acyclic}
    \OPTs(I) = \min\set{ \ \sum_{i=1}^{k-1} \beta(e_i, e_{i+1}) : k \in \N, e_1,\dots, e_k \in \El, e_s = e_1 \prec e_2 \prec \dots \prec e_k = e_t}.
    \end{align}
    In the following, we derive a recursive formula for determining the optimal value.
    Let $m_\ell = \lvert \El \rvert$ and order the edge set $\El = \fromto{e_1}{e_{m_\ell}}$ such that we have $i < j$ for all pairs $e_i, e_j \in \El$ with $e_i \prec e_j$.
    Such an order always exists, and it can be derived from a topological vertex ordering in $G$,
    like in the proof of \cref{lem:BSP-strong-follower-acyclic}.
    Note that $e_s \prec e_t$ holds because we assume that the problem is feasible.
    Moreover, we may assume that, for all edges $e \in \El \setminus \set{e_s, e_t}$, we have $e_s \prec e \prec e_t$.
    Indeed, edges $e$ for which this is not satisfied can be neglected because they do not appear in any feasible leader's solution.
    Hence, we may assume that $e_1 = e_s$ and $e_{m_\ell} = e_t$.
    
    Now the following is a valid recursive definition that can be evaluated for all $e \in \El$ in the order $e_1, \dots, e_{m_\ell}$:
    \begin{align*}
        S[e_s] &= 0 \\
        S[e'] &= \min\set{S[e] + \beta(e,e') : e \in \El, e \neq e', e \prec e'} &\forall e' \in \El \setminus \set{e_s}
    \end{align*}
    It can be shown by induction that $\OPTs(I) = S[e_t]$, corresponding to the formulation~\cref{eq:BSP-strong-acyclic}.
    Furthermore, since all values $\beta(e,e')$, as well as the ordering $e_1,\dots,e_{m_\ell}$
    and the evaluation of the recursive formula, can be computed efficiently, we conclude that $\OPTs(I)$ can be computed in polynomial time.
    By keeping track of the selected $e \prec e'$ when evaluating $S[e']$, an optimal leader's solution can be constructed explicitly as well.
\end{proof}

\section{Case of few leader's edges}
\label{sec:few-leader-edges}

Since the previous sections revealed many intractability results, it is natural to consider restrictions of the {\BSP} problem that render it tractable. In this section, we consider the restriction where the input instance has only a constant amount of leader's edges, i.e.\ $\lvert \El \rvert = O(1)$.

We show that, in this case, the weak path completion variant ({\BSPweakUndir} and {\BSPweakDir}) can be solved in deterministic polynomial time. 
Furthermore, the strong path completion variant in undirected graphs ({\BSPstrongUndir}) can be solved in randomized polynomial time, assuming polynomially bounded leader's and follower's edge costs. (The case of superpolynomial edge costs remains an open problem.)
The latter result is obtained by showing that the problem is equivalent to a known problem, the so-called \kcycle{k} problem. 
The algorithm can be derandomized if and only if the \kcycle{k} problem can be solved in deterministic polynomial time.
Finally, in case of strong path completion and directed graphs ({\BSPstrongDir}), we show that the restriction $\lvert \El \rvert = O(1)$ does not help, by showing that the problem is NP-complete already for $\lvert \El \rvert = 1$.

We now present the results for {\BSPweakUndir}, {\BSPweakDir}, and {\BSPstrongDir}.
The connection between {\BSPstrongUndir} and the \kcycle{k} problem is then developed in \cref{sec:k-cycle}.

\begin{observation}
\label{obs:few-leader-edges-weak}
The problems {\BSPweakUndir} and {\BSPweakDir} can be solved in $O(2^{\lvert \El \rvert} \cdot \mathrm{poly}(\lvert V \rvert))$ time. In particular, if $\lvert \El \rvert = O(1)$, they can be solved in polynomial time.
\end{observation}
\begin{proof}
    For both problems, the follower's problem can be solved in polynomial time by \cref{lem:weak-followers-problem}.
    Hence, the leader can enumerate all $2^{\lvert \El \rvert}$ choices of leader's solutions, predict the corresponding follower's responses, and choose the best among them.
  \end{proof}

\begin{theorem}
  \label{thm:few-leader-edges-strong-NP-complete}
    The problem {\BSPstrongDir} is NP-complete even if $\lvert \El \rvert = 1$.
\end{theorem}
\begin{proof}
    By an enumeration argument similar to \cref{obs:few-leader-edges-weak}, the problem is contained in the class NP whenever $\lvert \El \rvert = O(1)$.
    For the NP-hardness proof, we consider the NP-complete {\VDP} problem \cite{fortune1980directed}. 
    In this problem, we are given a directed graph $G = (V, E)$ and four distinct vertices $s_1,t_1,s_2,t_2 \in V$. 
    The question is if there are an $s_1$-$t_1$-path~$P_1$ and an $s_2$-$t_2$-path $P_2$ in $G$ such that $P_1$ and $P_2$ are vertex-disjoint.
    
    Let an instance $(G, s_1,t_1,s_2,t_2)$ of the {\VDP} problem be given.
    We can w.l.o.g.\ assume that $G$ contains neither of the two edges $(s_1, t_2)$ and $(t_1, s_2)$. 
    We define an instance of {\BSPstrongDir} as follows. 
    All edges $e$ of $G$ become follower's edges with $c(e) = 0$ and $d(e) = 1$. 
    We also insert an additional follower's edge $e = (s_1, t_2)$ with leader's cost $c(e) = 1$ and follower's cost $d(e) = 0$.
    Finally, we insert a single leader's edge $e = (t_1, s_2)$ with leader's cost $c(e) = 0$ and follower's cost $d(e) = 0$.
    We define the source vertex as $s := s_1$ and the sink vertex as $t := t_2$.
    This completes the description of the instance $I$.
    
    We claim that $\OPTs(I) = 0$ if and only if there are an $s_1$-$t_1$-path and an $s_2$-$t_2$-path in $G$ that are vertex-disjoint.
    Indeed, since $\lvert \El \rvert = 1$, the leader has only two choices: 
    If the leader's solution is $X = \emptyset$, then the follower uses the edge $(s_1, t_2)$ to go directly from the source to the sink, since this is the only path of follower's cost 0. This causes cost 1 for the leader.
    If the leader's solution is $X = \set{(t_1, s_2)}$, there are two cases: 
    If there are no vertex-disjoint $s_1$-$t_1$- and $s_2$-$t_2$-paths in $G$, then the follower's problem is infeasible. Hence, this choice for the leader is also infeasible. 
    If there are disjoint paths, then the follower is forced to use the edge $(t_1, s_2)$ and connect to it using an $s_1$-$t_1$-path and an $s_2$-$t_2$-path that are vertex-disjoint. Hence, they cannot use the direct edge $(s_1, t_2)$, which is the only edge of nonzero leader's cost. Hence, the leader has cost 0 in this case.
    In summary, we have shown that $\OPTs(I) = 0$ if and only if the given instance of the {\VDP} problem is a yes-instance.
  \end{proof}

  Observe that the proof of \cref{thm:few-leader-edges-strong-NP-complete} implicitly contains an NP-hardness proof for the follower's problem of {\BSPstrongDir} as well.
  In fact, when considering the instance constructed in this proof and assuming that the leader chooses the solution $X = \El = \set{(t_1, s_2)}$, the task of deciding whether two vertex-disjoint paths exist is to be solved by the follower.
  More precisely, the follower's problem $\ProblemFOLs(I, \El)$ is feasible if and only if the given instance of the {\VDP} problem is a yes-instance.
  This gives the following result:

  \begin{corollary}
    The follower's problem of {\BSPstrongDir} is NP-complete even if $\lvert \El \rvert = 1$.
  \end{corollary}
  
  Note that this statement strengthens the NP-completeness result of \cref{thm:BSP-strong-follower-NP-complete}.
  However, the second part of \cref{thm:BSP-strong-follower-NP-complete} does not apply here.
  Indeed, under the assumption that the given leader's solution $X$ is optimal, informally speaking, the leader has already solved the NP-complete {\VDP} problem and the follower does not have to do so again.
  However, as explained before \cref{thm:BSP-strong-follower-NP-complete}, this assumption cannot be tested in polynomial time and therefore does not lead to a well-defined problem (in NP).

  \subsection[Equivalence to the shortest-k-cycle problem]{Equivalence to the shortest-$k$-cycle problem}
  \label{sec:k-cycle}

  In this subsection, we turn our attention to the problem variant {\BSPstrongUndir}.
  Interestingly, it turns out that this variant of our problem is equivalent to the so-called \kcycle{k} problem, whose complexity status is still not entirely understood by the research community.

\begin{quote}
    \textbf{Problem} \kcycle{k}

    \textbf{Input:} An undirected graph $G = (V, E)$, a vertex subset $K \subseteq V$ of size $\lvert K \rvert = k$, edge weights $w : E \to \R_{\geq 0}$, and a threshold $T \in \R_{\geq 0}$.
    
    \textbf{Question:} Is there a cycle $C$ that includes all vertices in $K$ and has weight $w(C) \leq T$?
\end{quote}

Note that we introduced the \kcycle{k} problem in its decision version (involving the threshold parameter $T$), since we also did so for {\BSPstrongUndir}. For $k = \Omega(n)$, it is easily seen that the traveling salesperson problem reduces to \kcycle{k}, and so the latter is NP-hard. However, for $k = O(1)$, the \kcycle{k} problem has a very interesting, partially unresolved complexity status.
\begin{itemize}
\item The \emph{feasibility} question, which simply asks whether a cycle $C$ including all vertices in~$K$ exists, can be solved in deterministic polynomial time $2^{O(k^c)}n^2$
  for some large constant~$c$~\cite{kawarabayashi2008improved}. However, a speedup to $2^kn^{O(1)}$ is possible using randomization \cite{bjorklund2012shortest}.
    \item The \emph{optimization} variant with \emph{polynomially bounded} integer weights has been shown to admit a polynomial-time randomized algorithm \cite{bjorklund2012shortest} if $k$ is constant. (Note that \cite{bjorklund2012shortest} deals only with unit weights, but by a standard edge subdivision process, this result can be lifted to polynomially bounded weights.) It is currently an open question whether a derandomization is possible.
    \item The most general question, solving the \kcycle{k} problem with arbitrary weights is neither known to be NP-hard, nor to admit a polynomial-time (even randomized) algorithm.
    \item In the two special cases where $k = 1$ or $k=2$,
      the \kcycle{k} problem in its most general form
      admits a deterministic polynomial-time algorithm \cite{bjorklund2012shortest}.
      We note that \cite{bjorklund2012shortest} claims that \cite{nealOpen} contains an argument for $k=3$ as well, but in our opinion \cite{nealOpen} does not support such a claim.
\end{itemize}

We now obtain the following equivalence result between the two stated problems. (We remark that we make the assumption $k' \geq 3$, but this is not a significant restriction because, for $k' \in \set{1, 2}$, a polynomial-time deterministic algorithm for \kcycle{k'} is known anyway.)

\begin{theorem}
    \label{thm:equivalence-k-cycle}
     Let $k,k' \in \N$, $k' \geq 3$. The problem {\BSPstrongUndir} with $\lvert \El \rvert = k$ is equivalent to the \kcycle{k'} problem in the precise following way:
    \begin{itemize}
        \item A given instance $I$ of {\BSPstrongUndir} with $\lvert \El \rvert = k$ can be solved in polynomial deterministic time and $O(2^k\log \lvert I \rvert )$ oracle calls to \kcycle{(k+1)}.
        \item A given instance of \kcycle{k'} can be solved in polynomial deterministic time and a single oracle call to {\BSPstrongUndir} with $\lvert \El \rvert = 2k' - 2$.
    \end{itemize}
    For both of the above reductions, if additionally the input instance has polynomially bounded integer weights, then the instances produced by the oracle calls have the same property.
  \end{theorem}

In particular, \cref{thm:equivalence-k-cycle} has the following implications: For polynomially bounded edge weights, the problem {\BSPstrongUndir} with $\lvert \El \rvert = O(1)$ has a polynomial-time randomized algorithm, 
and this algorithm can be derandomized if and only if the \kcycle{k} problem can be solved in deterministic polynomial time (for $k = O(1)$ and polynomially bounded edge weights).
Furthermore, for general edge weights and $k,\lvert \El \rvert = O(1)$, {\BSPstrongUndir} has an efficient algorithm if and only if \kcycle{k} has one.
In particular, a deterministic polynomial-time algorithm for {\BSPstrongUndir} with $\lvert \El \rvert = 1$ and general edge weights results from the deterministic polynomial-time algorithm for the \kcycle{k} problem with $k = 2$.

\begin{proof}[Proof of \cref{thm:equivalence-k-cycle}]

    We first show the reduction {\BSPstrongUndir} $\to$ \kcycle{k}. 
    Given an instance $I = (G, \El, \Ef, s, t, c, d, T)$ of {\BSPstrongUndir} with $G = (V,E), \lvert \El \rvert = k$ and threshold $T \in \R_{\geq 0}$, we want to decide whether $\OPTs(I) \leq T$. This can be achieved in the following way: 
    We first insert a dummy vertex $v_d$ and edges $\set{v_d, s}$ and $\set{v_d, t}$ into the graph. 
    This means that any cycle that includes $v_d$ also traverses the vertices $s$ and $t$, using the two new edges.
    Now we iterate over all of the $2^k$ leader's choices $X \subseteq \El$, and we predict how much leader's cost the follower causes in response. This can be done as follows:
    We define the graph $G'$ with vertex set $V(G') :=  V \cup \set{v_d}$, and edge set $E(G') := \Ef \cup X \cup \set{\set{v_d, s}, \set{v_d, t}}$.
    We choose a large integer constant $M = 1 + \left\lceil \sum_{e\in E}c(e) / \min_{e \in E, d(e) \neq 0} d(e) \right\rceil$. We define edge weights on $G'$ as $w(e) = Md(e) + c(e)$
    for $e \in \Ef \cup X$ and $w(e) = 0$ for $e \in \set{\set{v_d, s}, \set{v_d, t}}$. We now consider the \kcycle{(k+1)} problem on $G'$ for the set $K := \set{v_d} \cup X$ of vertices and edges.\footnote{Note that, in the definition of the \kcycle{k} problem, the set $K \subseteq V$ is a set of vertices. However, we can w.l.o.g.\ assume that $K \subseteq V \cup E$. Indeed, for some edge $e \in E \cap K$, we can subdivide $e$ into two new edges $e_1$ and $e_2$, introducing a new vertex $v$ incident to them. We can then replace $e$ in $K$ by $v$ and modify the weights such that $w(e_1) + w(e_2) = w(e)$. Note that this does not change $\lvert K \rvert$.}
    We now invoke an oracle for solving this instance of the \kcycle{(k+1)} problem.
    Using a binary search with at most $O(\log \lvert I \rvert)$ iterations,
     we can find the minimal value~$W^\star$ of $w(C)$ over all cycles $C$ in $G'$ that include $K$. 
    Note that, due to the definition of $G'$, any such cycle consists of the dummy vertex~$v_d$, its incident edges $\set{t, v_d}$ and $\set{v_d, s}$, and an $s$-$t$-path that traverses all of $X$, but no edge of $\El \setminus X$. 
    In fact, it can be seen that there is a one-to-one correspondence between such cycles and feasible solutions to the follower's problem $\ProblemFOLs(I, X)$.
    Moreover, due to the choice of $M$, we have $W^\star = M \OptFOLs(I) + c^\star$, where $c^\star$ is the leader's cost caused by the optimal follower's solution. 
    Hence, we can determine the latter as $c^\star = (W^\star \bmod M)$. By iterating over all of the $2^k$ choices of $X \subseteq \El$, we can find the minimal possible value of $c(X) + c^\star$. This gives the optimal leader's cost and thus solves the problem {\BSPstrongUndir}. Finally note that, if the input instance has only polynomially bounded integer costs $c,d$, the new instance also has this property by our definition of the weights $w$.

    We now proceed with the other direction \kcycle{k} $\to$ {\BSPstrongUndir}. This reduction basically follows from the vertex fixing lemma (\cref{lem:lifting-lemma}). 
    Given an instance $(G, K, w, T)$ of \kcycle{k'}, with a graph $G = (V, E)$, a subset $K \subseteq V$ with $\lvert K \rvert = k'$, weights $w : E \to \R_{\geq 0}$, and a threshold~$T \in \R_{\geq 0}$, the question is to decide whether there is a cycle $C$ that includes $K$ and has weight $w(C) \leq T$. 
    We modify the graph $G$ in the following way: We choose an arbitrary vertex $v \in K$. 
    We create a twin $v'$ of $v$, i.e.\ we introduce a new vertex $v'$ that has the same neighborhood as $v$ (but $v$ and $v'$ are not connected). Let $G'$ be the resulting graph.
    We define edge costs $c$ in $G'$ by $c(e) := w(e)$ for $e \in E(G) \cap E(G')$, i.e.\ for all edges that are not incident to the new vertex $v'$, and by letting $c(\set{v', u}) := w(\set{v, u})$ for the remaining edges $\set{v', u}$.
    Moreover, we define $s := v$ and $t := v'$. 
    Let $W := K \setminus \set{v}$.
    We claim that cycles $C$ in the old graph $G$ that include $K$ correspond exactly to  $s$-$t$-paths $P$ in the new graph $G'$ that include $W$. 
    Indeed, any such cycle $C$ in the old graph can be transformed into a corresponding path $P$ by introducing the twin vertices again.
    Any such path $P$ can be transformed into a cycle~$C$ by contracting $v$ and $v'$ into a single vertex. 
    Note that, in the special case $\lvert E(P) \rvert = 2$, this contraction does not create a simple cycle. 
    However, the case $\lvert E(P) \rvert = 2$ is excluded due to $k' \geq 3$, which implies $\lvert W \rvert \geq 2$.
    Furthermore, the weight $w(C)$ of such a cycle is equal to $c(P)$ for the corresponding path.
    We turn $G'$ into an instance $I = (G', \El, \Ef s,t, c, d)$ of {\BSPstrongUndir} by letting $\El = \emptyset$, $\Ef = E(G')$, $s,t,c$ as defined above, and $d := c$.
    Finally, we apply the vertex fixing lemma (\cref{lem:lifting-lemma}) to the instance $I$ and the vertex set $W$ in order to obtain a new instance $I'$.
    Then there are the following two cases, as given in the vertex fixing lemma.
    In the first case,
    we have $\OPTs(I') \geq M$ for some large number $M$.
    This means that there is no $s$-$t$-path in $G'$ that includes $W$ and, accordingly, also
    the given instance of the \kcycle{k} problem is infeasible.
    In the second case, we have  $\OPTs(I') < M$ and, by the vertex fixing lemma, the value $\OPTs(I')$ is equal to the optimal value of problem~(\ref{eq:BSP-modified}). Due to $\El = \emptyset$ and $c = d$, this means that
    \begin{align*}
        \OPTs(I') &= \min \set{c(P) : P \text{ is an $s$-$t$-path in $G'$ that includes $W$}}\\
                &= \min \set{w(C) : C \text{ is a cycle in $G$ that includes $K$}}.
    \end{align*}
    Therefore, we can solve the \kcycle{k'} problem with a single oracle call to {\BSPstrongUndir},
    in which we use the same threshold~$T$ as the desired one for the given \kcycle{k'} instance.
    Finally, we remark that, in the instance $I'$, we have at most $2\lvert W \rvert + 4 \lvert \El \rvert = 2k' - 2$
    leader's edges due to \cref{lem:lifting-lemma}.
    Furthermore, it follows from the construction presented in \cref{lem:lifting-lemma}
    that, if the weights $w$ are integral and polynomially bounded,
    so are the costs in the new instance $I'$.
    (Note that the parameter $\varepsilon$ in the proof of \cref{lem:lifting-lemma} can be set as $\varepsilon = 1$ here.)
\end{proof}

\section{Inapproximability}
\label{sec:inapproximability}
In this section, we strengthen the previous hardness results obtained in \cref{thm:BSP-strong-follower-NP-complete,thm:BSP-strong-sigma-2-complete,thm:few-leader-edges-strong-NP-complete} 
by showing that all these problems are not only NP-hard ($\Sigma^p_2$-hard), but they are even inapproximable.
By \emph{inapproximable}, we mean that there is no polynomial-time approximation algorithm with a subexponential approximation factor, unless $\mathrm{P} = \mathrm{NP}$.

\begin{theorem}
  The following problems are all inapproximable:
  \begin{itemize}
  \item the strong path completion variant of the {\BSP} problem, both in directed and in undirected graphs (\cref{thm:BSP-strong-sigma-2-complete})
  \item the follower's problem of the strong path completion variant of the {\BSP} problem, both in directed and in undirected graphs (\cref{thm:BSP-strong-follower-NP-complete})
  \end{itemize}
  For the leader's problem in directed graphs, this holds even in the special case $\lvert \El \rvert = 1$ (\cref{thm:few-leader-edges-strong-NP-complete}).
\end{theorem}
\begin{proof}
    In each of the \cref{thm:BSP-strong-sigma-2-complete,thm:few-leader-edges-strong-NP-complete},
    the instances that result from the hardness proof have the property that their optimal value is either at least 1 or equal to 0.
    Therefore, it is NP-hard to distinguish between such instances of optimal value 0 and instances of optimal value at least~1, which implies that the problems cannot be approximated at all.
    Note that, in the case of \cref{thm:BSP-strong-sigma-2-complete}, this distinction is even $\Sigma^p_2$-hard, which is at least as hard as NP-hard.

    In \cref{thm:BSP-strong-follower-NP-complete}, the situation is very similar. Here, the instance constructed in the hardness proof has the property that its optimal value is either $\varepsilon$ or at least $1 + \varepsilon$, where $\varepsilon > 0$ is an arbitrarily small constant (i.e.\ exponentially small if $\varepsilon$ is encoded in the input).
    This implies that an approximation algorithm whose approximation factor is subexponential in the input size is impossible, unless $\mathrm{P} = \mathrm{NP}$.
\end{proof}

Finally, for the weak path completion variant of the {\BSP} problem (which has been shown to be NP-hard in \cref{thm:BSP-weak-NP-complete}), we have no inapproximability conclusion, since, in this case, our reduction does not result in instances with optimal values 0 and at least 1. It remains an open question whether an approximation of any kind is possible for the weak path completion variant of the {\BSP} problem.

\section{Conclusion}
\label{sec:conclusion}

We introduced the partitioned-items {\BSP} problem, where some edges of a graph are controlled by the leader and the others by the follower. We considered the leader's problem and the follower's problem, each in the strong and in the weak path completion variant, and each restricted to undirected, directed, or directed acyclic graphs. 
We completely characterized each of the corresponding (decision) problems as either solvable in polynomial time, NP-complete, or even $\Sigma^p_2$-complete. 
In particular, we proved that the (undirected or directed) {\BSP} problem in the strong path completion variant is $\Sigma^p_2$-complete, based on a newly introduced min-max version of the {\HamPath} problem.
All of our hardness results for the strong path completion variant also imply inapproximability results. In case of the weak path completion variant, it remains an open question whether efficient approximation is possible.

On the positive side, using a dynamic programming approach, it is possible to solve the {\BSP} problem in the strong path completion variant in polynomial time, when restricted to directed acyclic graphs.
Furthermore, in the special case of $\lvert \El \rvert = O(1)$ and polynomially bounded edge weights, 
we showed that an equivalence to the \kcycle{k} problem can be utilized to solve the undirected {\BSP} problem with strong path completion in randomized polynomial time.
It remains an open question whether a deterministic polynomial-time algorithm exists in this case. It is also an open question what happens in the same setting, 
but when arbitrary edge weights are permitted. By our equivalence result, 
this question is equivalent to the analogous one about the \kcycle{k} problem.
One could also study our problem for fixed small numbers $\lvert \El \rvert \geq 2$,
which could help understanding the \kcycle{k} problem for fixed small $k$.
Finally, we remark that our equivalence of the \kcycle{k} problem and our problem for $k, \lvert \El \rvert = O(1)$ (\cref{thm:equivalence-k-cycle}) introduces an additional factor in both directions (one direction performs $2^k$ oracle calls, 
the other roughly doubles the parameter $k$).
It is unknown whether a more efficient reduction exists, which would show the equivalence between these two problems in an even stricter fashion. 

\printbibliography

\end{document}